\documentclass{article}

\usepackage[preprint]{neurips_2025}

\usepackage[utf8]{inputenc} 
\usepackage[T1]{fontenc}    
\usepackage{hyperref}       
\usepackage{url}            
\usepackage{booktabs}       
\usepackage{amsfonts}       
\usepackage{nicefrac}       
\usepackage{microtype}      
\usepackage{xcolor}         

\usepackage{multirow}
\usepackage{graphicx}
\usepackage{float}
\usepackage{subfig}
\usepackage{enumerate}
\usepackage{algorithm}
\usepackage{algpseudocode}

\usepackage{amsmath}
\usepackage{amssymb}
\usepackage{mathtools}
\usepackage{amsthm}
\usepackage{bbm}

\graphicspath{{figure/}}
\DeclareGraphicsExtensions{.pdf,.jpeg,.png}

\theoremstyle{plain}
\newtheorem{theorem}{Theorem}[section]
\newtheorem{proposition}[theorem]{Proposition}
\newtheorem{lemma}[theorem]{Lemma}
\newtheorem{corollary}[theorem]{Corollary}
\theoremstyle{definition}

\theoremstyle{remark}
\newtheorem{remark}[theorem]{Remark}

\DeclareMathOperator{\Tr}{Tr}
\DeclareMathOperator{\diag}{diag}

\algnewcommand{\IfThenElse}[3]{
	\State \algorithmicif\ #1\ \algorithmicthen\ #2\ \algorithmicelse\ #3}

\title{Recovering Fairness Directly from Modularity: a New Way for Fair Community Partitioning}

\author{%
  Yufeng Wang\\
  School of Mathematics and Statistics\\
  Xidian University\\
  Xi'an, China \\
  \And
  Yiguang Bai\thanks{Corresponding authors.}\\
  School of Mathematics and Statistics\\
  Xidian University\\
  Xi'an, China \\
  \texttt{ygbai@xidian.edu.cn} \\
  \And
  Tianqing Zhu\\
  City Faculty of Data Science\\
  University of Macau\\
  Macau, China
  \And
  Ismail Ben Ayed\\
  École de technologie supérieure\\
  Université du Québec
  \And
  Jing Yuan\\
  College of Mathematical Medicine\\
  Zhejiang Normal University\\
  Jinhua, China
}

\begin{document}

\maketitle

\begin{abstract}
Community partitioning is crucial in network analysis, with modularity optimization being the prevailing technique. However, traditional modularity-based methods often overlook fairness, a critical aspect in real-world applications. To address this, we introduce protected group networks and propose a novel fairness-modularity metric. This metric extends traditional modularity by explicitly incorporating fairness, and we prove that minimizing it yields naturally fair partitions for protected groups while maintaining theoretical soundness. We develop a general optimization framework for fairness partitioning and design the efficient Fair Fast Newman (FairFN) algorithm, enhancing the Fast Newman (FN) method to optimize both modularity and fairness. Experiments show FairFN achieves significantly improved fairness and high-quality partitions compared to state-of-the-art methods, especially on unbalanced datasets.
\end{abstract}



\section{Introduction}

Community detection is a crucial scientific problem in network analysis, essential for understanding the structure and dynamics of complex networks. Related algorithms for community detection in a complex network are extensively utilized in various fields, including social networks, recommendation systems, bioinformatics, and traffic networks~\citep{fortunato2010community}.
Traditional algorithms primarily focus on maximizing community cohesiveness, often using modularity, a widely adopted metric introduced by \citet{fn_algorithm}, as a measure to evaluate the quality of the resulting partitions given its close correspondence with the intrinsic structural properties observed in real-world networks~\citep{fortunato202220}. 
This has led to the development of numerous popular algorithms based on modularity optimization, including the Fast Newman algorithm \citep{fn_algorithm}, the Clauset-Newman-Moore algorithm \citep{cnm_algorithm}, the Louvain algorithm \citep{FUA_algorithm}, and the fast complex network clustering algorithm \citep{FNCA_algorithm}. Other approaches include the random walk ant colony algorithm \citep{RWACO_algorithm}, the label propagation-based algorithm \citep{lpa_algorithm}, the intuitive-based Girvan-Newman algorithm \citep{gn_algorithm}, and information-based algorithm \citep{encoding_detection}.

While these traditional modularity-based community partitioning algorithms have achieved significant results in community detection,
a substantial limitation is that they often prioritize structural and cohesive community properties but omit fairness considerations. 
This can result in biased partitioning, especially when nodes are characterized by diverse attributes like gender, age, or race;
such biases lead to overrepresentation or underrepresentation of specific groups within the detected communities. 
The concept of fairness in community partitioning promotes the constraint that the distribution of relevant attributes within each identified community should match the distribution observed in the global network.
It is especially critical in sensitive application domains, such as social networks and criminal justice systems, where the deployment of biased algorithms has the potential to exacerbate existing societal inequalities.
The COMPAS algorithm, employed in the U.S. criminal justice system, exemplifies this issue, having been shown to exhibit racial bias in sentencing decisions~\citep{COMPAS}. 
Similar issues arise in other contexts; for example, biased power distribution may result in significant inequalities if regional needs are not addressed equitably, particularly during periods of resource scarcity. 
To address fairness in community partitioning, researchers have explored various approaches. 
\citet{FairSC_fair_algorithm} pioneered the concept of fair clustering, introducing "fairlets" to transform the problem into a more traditional clustering task and proposing approximation algorithms along with a "balance" metric for fairness. \citet{VFC_fair_algorithm} used Kullback-Leibler (KL) divergence to measure community fairness, incorporating it as a penalty in objective functions for algorithms like K-Means, K-Medians, and Ncut.  
\citet{bfkm_algorithm} similarly studied fairness in K-Means, adding a penalty based on the difference between the protected group proportions within a community and the overall population. And then optimized it using coordinate descent.
For spectral clustering, \citet{pmlr-v97-kleindessner19b} used linear constraints based on the balance metric, and \citet{pmlr-v206-wang23h} later proposed the s-FairSC algorithm to improve the computational efficiency of Kleindessner's FairSC approach. \citet{multi_view} developed a fair multi-view clustering algorithm, using the degree of protected group nodes as a penalty within the Rcut framework. \citet{fce_algorithm} achieved fair ensemble clustering through the vector rotation and the penalty term constructed by the Frobenius norm, which also ensured that the sizes of clusters are similar. Finally, \citet{10.1145/3625007.3627518} connected fair community partitioning with the concept of "modularity fairness."

A critical challenge of enforcing fairness in community partitioning is the potential trade-off between the effect of maximizing communities' internal cohesion and the achievement of fairness in each community. This necessitates identifying a suitable balance among these competing objectives.
While several studies introduced metrics for evaluating fairness in the context of modularity, no existing algorithm incorporates fairness directly within the modularity optimization procedure.
This motivates the present study, which proposes an integrated approach to community partitioning by incorporating fairness considerations within the established framework of modularity-based methods, thereby providing a more comprehensive and elegant approach to the problem of fair community partitioning.

Our main contributions to fair community partitioning in this study can be summarized as follows. Firstly, we introduce the concept of \emph{protected group network}, which presents a novel perspective for analyzing fairness in community partitioning. This perspective facilitates a deeper understanding of community fairness, which lays the groundwork in theory for fair community detection and derives the novel \emph{fairness-modularity} ($Q^P$) based on the modularity of the introduced protected group network. Particularly, we show that minimizing $Q^P$ directly leads to fair community partitioning. 
Moreover, we propose a new modularity optimization algorithm, the so-called FairFN algorithm, to achieve efficient and effective fair community partitioning. It pioneers a way to integrate fairness directly into modularity optimization without significantly increasing computational complexity. 
The proposed FairFN algorithm outperforms state-of-the-art fair clustering algorithms in fairness while obtaining better community partitioning results, even on unbalanced datasets. It presents a novel groundbreaking solution for fair community detection in both theory and numerics.

\section{Preliminaries}
{\bf Modularity of undirected graph:} Let $G=(V,E)$ be an undirected graph, with vertex set $V=\{v_1,v_2,\cdots,v_n\}$ and edge set $E$, and let $A = (A_{ij})\in R^{n\times n}$ denote the adjacency matrix of $G$. The objective of community partitioning is to divide $G$ into a set $\mathbb{C} =\{C_1,C_2,\cdots,C_k\}$ of $k$ communities, where $\bigcup_{u=1}^k C_u=V$, $C_i\cap C_j=\emptyset$ for any $i$ and $j$, and all $C_u$ are nonempty.

Let $k_i$ denote the degree of vertex $i$, and $K=(k_1,k_2,\cdots,k_n)^T$ the degree vector. $m=(\sum_i k_i)/2$ is thus the total number of edges in $G$.
The modularity $Q$ to evaluate the effect of a community partitioning $S=(S_{iu})\in R^{n\times k}$ is defined as \citep{gn_algorithm}:
\begin{equation}
    Q=\frac{1}{2m}\sum_{ij}\sum_{u} \left[A_{ij}-\frac{k_ik_j}{2m}\right]S_{iu}S_{ju}, \label{eq_modularity_matrix}
\end{equation}
where $A=(A_{ij})$ is the adjacency matrix of $G$ and
\begin{equation}
    S_{iu}=\begin{cases}
        1, \,\,\text{if } v_i\in C_u\\
        0,\,\,\text{otherwise}
    \end{cases},
\end{equation}
for each $i=1,2,\cdots,n$ and $u=1,2,\cdots,k$. 
It can also be rewritten as follows \citep{newman2006modularity}:
\begin{equation}
    Q=\frac{1}{2m}\Tr \left(S^TBS\right),
\end{equation}
where $B=A-\frac{1}{2m}KK^T$.

Clearly, the modularity metric \eqref{eq_modularity_matrix} measures the quality of a community partition by comparing the observed edge density within the identified communities to the expected density in a randomized network configuration model. A higher modularity score indicates that the identified communities are more distinct, exhibiting a significantly bigger number of internal connections than expected by chance.

{\bf Modularity of directed graph:} Modularity metric \eqref{eq_modularity_matrix} can be generalized in a directed graph setting \citep{directed_modularity}. Let $K^{in}=(k^{in}_1,k^{in}_2,\cdots,k^{in}_n)^T$ be the in-degree vector and $K^{out}=(k^{out}_1,k^{out}_2,\cdots,k^{out}_n)^T$ the out-degree vector of the directed graph $G_d=(V,E)$. The number of edges is then $m^{directed}=\sum_i{k^{in}_i}+\sum_i{k^{out}_i}$. The modularity $Q$ of a community partitioning to the directed network  $G_d$ can be defined by
\begin{equation} \label{eq_modularity_directed}
    Q=\frac{1}{m^{directed}}\sum_{ij}\sum_{u} \left[A_{ij}-\frac{k^{in}_ik^{out}_j}{m^{directed}}\right]S_{iu}S_{ju}.
\end{equation}

\section{Method}

In this section, we first introduce the \emph{protected group network} $G^P$ for the given graph $G$ and its associated fairness-modularity $Q^P$. We prove that a fair community partition can be achieved by minimizing $Q^P$. Therefore, the fair community partitioning of $G$ boils down to a multi-objective optimization problem of maximizing the modularity of $G$ while minimizing the fairness-modularity of its associated protected group network $G^P$. 
Moreover, we provide a novel efficient algorithm to solve the introduced optimization problem of fair community partitioning.


{\bf Protected Group Network:} Fairness can be defined from a network perspective \citep{multi_view}. Now we adapt this perspective and begin by formally defining \emph{protected group network}.

Given the graph $G=(V,E)$, let $\mathbb{P}=\{P_1,P_2,\cdots,P_r\}$ be $r$ protected groups, where each $P_w$, $w=1,2,\cdots,r$, is a vertex subset of $V$ sharing a distinct attribute, such as gender or age, and 
\[
V=\bigcup_{w=1}^r P_w, \,\,P_i\cap P_j=\emptyset \quad \forall i,j.
\]
We define a \emph{protected group network} $G^P=(V,E^{P})$ as the union of $r$ complete directed graphs $(P_w, E^{P_w})$, where the edge set $E^{P_w}=\{(v_i,v_j)| v_i\in P_w,v_j\in P_w\}$ including self-loops, such that
\[
E^P=\bigcup_{w=1}^r E^{P_w}, \,\,E^{P_i}\cap E^{P_j}=\emptyset \quad \forall i,j.
\]

The protected group network needs self-loops, meaning that the diagonal elements of the corresponding adjacency matrix are 1.
For convenience, by re-indexing vertices, the adjacency matrix of $G^P$ can be represented as:
\begin{equation} \label{eq_dmatrix}
    D=\diag \left\{\mathbbm{1}_{|P_1|},\mathbbm{1}_{|P_2|},\cdots,\mathbbm{1}_{|P_r|}\right\},
\end{equation}
where $|P_w|$ denotes the number of vertices in $P_w$ and $\mathbbm{1}_{|P_w|}$ is a $|P_w|\times |P_w|$ matrix whose entries are all $1$. Clearly, $D$ is a block diagonal matrix consistent of $\mathbbm{1}_{|P_w|}$.

\begin{remark}
    The adjacency matrix of a graph remains structurally unchanged under re-indexing, as the adjacency matrix $A'$ after re-indexing is equivalent to the original adjacency matrix $A$ up to a permutation matrix $P$ such that $A'=PAP^T$.
    The value of modularity $Q$ also remains unchanged after re-indexing as
    \[
    Q=\frac{1}{2m}\Tr \left[(PS)^T(PBP^T)(SP)\right] = \frac{1}{2m}\Tr \left(S^TBS\right).
    \]
\end{remark}
Figure \ref{fig_pgnet} depicts an example of a protected group network, with three protected groups of vertices (in three different colors) including self-loops. 
Figure \ref{fig_NFP} represents a fair community partitioning result, where members of each protected group are mostly distributed evenly across communities. Figure \ref{fig_UP} shows the community partition missing fairness, which may lead to the segregation or underrepresentation of certain groups.

\begin{figure}[ht]
    \begin{center}
    \subfloat[nearly fair partitioning]{\includegraphics[width=.34\textwidth]{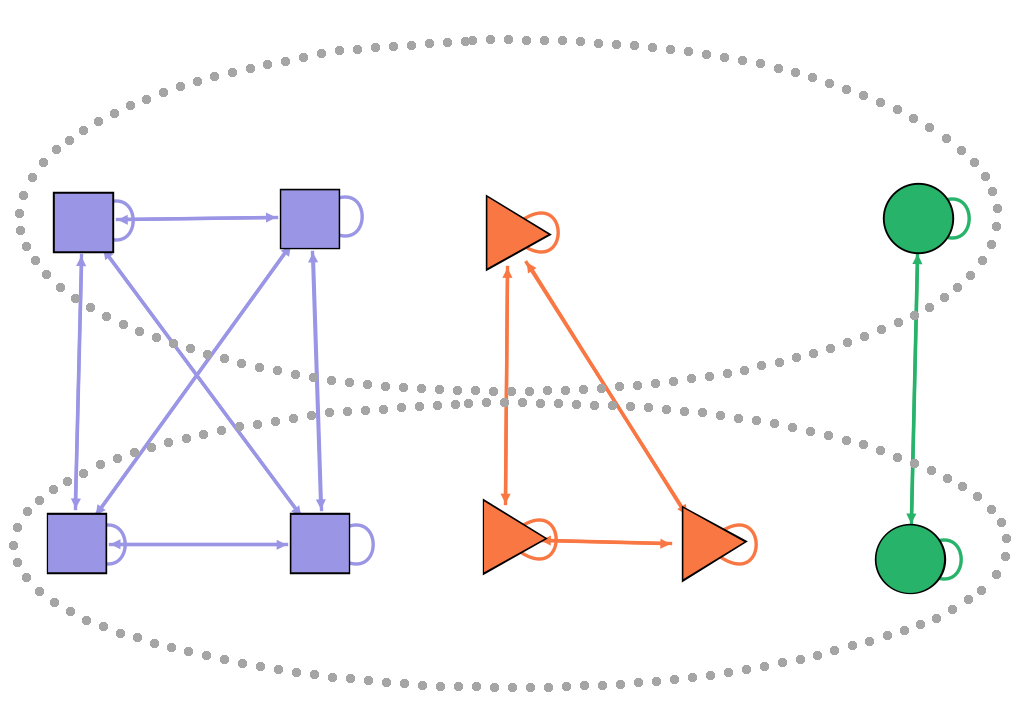}
    \label{fig_NFP}}
    \hspace{2cm}
    \subfloat[unfair partitioning]{\includegraphics[width=.34\textwidth]{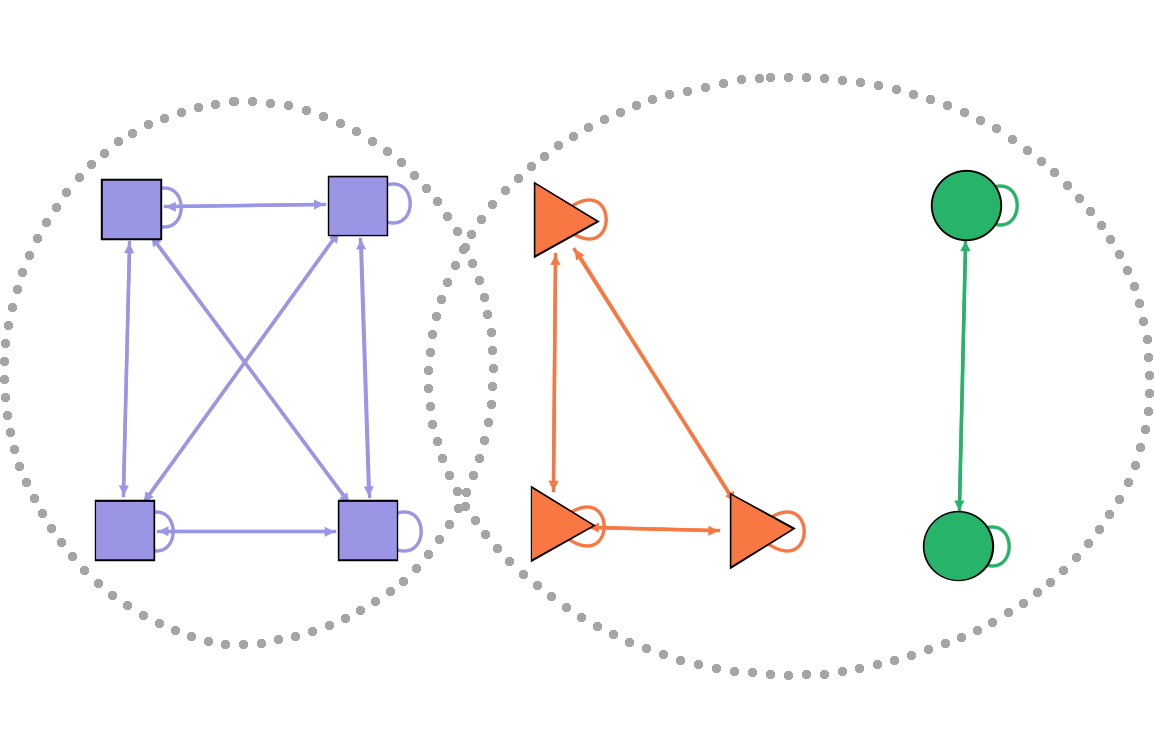}
    \label{fig_UP}}
    \caption{Partitioning examples of \emph{protected group network}. There are three protected groups, with the vertices in the same group represented by the same color. The vertices are divided into two communities represented by the dashed boxes.}
    \label{fig_pgnet}
    \end{center}
\end{figure}


{\bf Fairness-Modularity:} Actually, the protecctd group network $G^P = (V, E^P )$ is just a special case of directed graph $G_d$ with self-loops, such that $(v_j, v_i) \in E^P$ only when $(v_i, v_j) \in E^P$.
For such a directed graph $G_d$, we can easily prove the following result:

\begin{lemma}
    Let $G_d(V,E)$ be a directed graph such that $(v_j, v_i) \in E$ only when $(v_i, v_j) \in E$. And define $k_i=k^{in}_i=k^{out}_i$ for any node $v_i$ and $m=\sum_i{k^{in}_i}$. The modularity of its community partition $\mathbb{C} =\{C_1,C_2,\cdots,C_k\}$ can be calculated by
    \begin{equation}
        Q=\frac{1}{2m}\sum_{ij}\sum_{u} \left[A_{ij}-\frac{k_ik_j}{2m}\right]S_{iu}S_{ju},
    \end{equation}
    where $S$ is a community indicator matrix and $A$ is the adjacency matrix of $G_d$.
    \label{thm_corr}
\end{lemma}
\textbf{Proof:} See Appendix \ref{appendix_lemma_proof}.

By the above lemma, we see that an equivalent undirected graph can be constructed for the special directed graph $G_d=(V,E)$ with $(v_j, v_i) \in E$ only when $(v_i, v_j) \in E$, simply by combining the two directed edges $(v_i, v_j)$ and $(v_j, v_i)$ into one undirected edge for $i\neq j$ as well as the self-loop to the "0.5" number of an edge when considering degrees and the number of edges.

By Lemma \ref{thm_corr}, the modularity $Q^P$ of the protected group network $G^P=(V,E^P)$, for the community indicator matrix $S$ of $G^P$ which indicates $k$ communities $\mathbb{C}=\{C_1,C_2,\cdots,C_k\}$, can be calculated as follows:
\begin{equation}
        Q^P=\frac{1}{2m^P}\Tr\left(S^TB^PS\right)=\frac{1}{2m^P}\left[\Tr\left(S^TDS\right)-\frac{1}{2m^P}\Tr\left(S^TK^P(K^P)^TS\right)\right]
        \label{eq_QP},
\end{equation}
where $D$ is the adjacency matrix of $G^P$ given by \eqref{eq_dmatrix}, $B^P=D-\frac{1}{2m^P}K^P(K^P)^T$ and $K^P$ denotes the vector of each vertex degree (in-degree vector or out-degree vector) such that
\[
K^P = \big\{ \underbrace{|P_1|,\cdots,|P_1|}_{|P_1|\text{ elements}},\cdots,\underbrace{|P_r|,\cdots,|P_r|}_{|P_r|\text{ elements}}\big\} .
\]
$m^P=\frac12\sum_{w=1}^r|P_w|^2$ is hence the number of edges in the corresponding undirected graph of $G^P$.

In this work, we call $Q^P$ {\em fairness-modularity}. From the equation \eqref{eq_QP} of $Q^P$, the key result of fair community partitioning can be obtained:
\begin{theorem} \label{thm-fm}
    The modularity $Q^P$ of the community partition $\mathbb{C}=\{C_1,C_2,\cdots,C_k\}$ in the protected group network $G^P=(V,E^P)$ verifies $Q^P\ge 0$, and the equality $Q^P = 0$ holds if and only if the community partition $\mathbb{C}$ is fair, i.e.
    \begin{equation}
        \frac{|C_u\cap P_w|}{|C_u|}=\frac{|P_w|}{n}
    \end{equation}
    for any $u = 1,\cdots,k$ and $w=1,\cdots,r$.
\end{theorem}
\textbf{Proof:} See Appendix \ref{appendix_theorem_proof}.

Theorem \ref{thm-fm} indicates that a community partition that minimizes {\em fairness-modularity} $Q^P$ is fair. 
Therefore, finding a fair community partition of the given $G$ can be formulated as a multi-objective
optimization problem of maximizing the modularity of $G$
while minimizing the fairness-modularity of its associate
protected group network $G^P$, i.e.
\begin{equation} \label{multi-opt}
\max_{S} \;\; \{ Q, \, - Q^P \}
\end{equation}
where, the modularity $Q$ of $G$ for the community indicator
matrix $S$ is given in \eqref{eq_modularity_matrix} and the fairness-modularity $Q^P$ of $G^P$ is given in \eqref{eq_QP}. For convenience, we call $G$ the observed network.

The properties of $Q^P$ are given in Appendix \ref{appendix_property}.

{\bf General Framework and Fair Fast Newman Algorithm:} To address the new modularity-based fair partitioning problem \eqref{multi-opt}, we present a simple, general framework that demonstrates how to apply $Q^P$ to fair community partitioning. This approach requires only minimal additional judgment compared to traditional modularity optimization. The specific steps are as follows:

\begin{enumerate}[I.]
    \item \textbf{Initialization:} Initialize communities $\mathbb{C}$ which may come from random assignment, a community with a single vertex, or other rules.
    \item \textbf{Increment update:} Merge or move some vertices or communities and find a strategy that achieves the best increment of the modularity after merging or moving them.\label{em_general2}
    \item \textbf{Convergence judgment:} Repeat Step \ref{em_general2} until the increment of the modularity is nonpositive or satisfies other conditions.
\end{enumerate}
To consider fairness in the modularity optimization, we add an extra judgment in the Increment update stage. The algorithm only chooses the strategy that keeps the increment of the fairness-modularity negative after modification. Convergence judgment can be slightly loosened to get better fairness. Unlike previous fairness learning methods, our new general framework is based on modularity optimization. An obvious advantage is that this framework is applicable to most popular modularity optimization algorithms, such as the Fast Newman algorithm, the Louvain algorithm, etc.

Under the general modularity optimization framework, we propose a novel algorithm according to the Fast Newman (FN) algorithm, the so-called Fair Fast Newman (FairFN) algorithm, to tackle the introduced multi-objective optimization problem \eqref{multi-opt} of seeking a fair community partition of the given graph $G$. 

\begin{algorithm}[tb]
    \caption{Fair Fast Newman Algorithm \protect\footnotemark}
    \label{alg_fair_fn}
    \begin{algorithmic}[1]
    \Require Degree vector $K$ and edges number $m$ of $G$, $K^P$ and $m^P$ defined in \eqref{eq_QP}, parameter $\alpha$.
    \Ensure The partition $\mathbb{C}=\{C_1,C_2,\cdots,C_k\}$.
    \State Initialize $n$ communities $\mathbb{C}=\{C_1,C_2,\cdots,C_n\}$, each containing one vertex.
    \For{each pair of vertices $v_i$ and $v_j$}
        \IfThenElse{$v_i$ and $v_j$ are linked in $G$}{$e_{ij}\gets\frac{1}{2m}$}{$e_{ij}\gets0$}
        \IfThenElse{$v_i$ and $v_j$ are linked in $G^P$}{$e^P_{ij}\gets\frac{1}{2m^P}$}{$e^P_{ij}\gets0$}
        \State $a_i \gets \frac{K_i}{2m},a^P_i \gets \frac{K^P_i}{2m^P}$
    \EndFor
    \While{not all communities are merged into one and $\Delta Q_{ij} > -\frac{\alpha}{2m}$ for any pair $C_i$ and $C_j$}
        \State Compute $\Delta Q_{ij} = 2(e_{ij} - a_i a_j)$ and $\Delta Q^P_{ij} = 2(e^P_{ij} - a^P_i a^P_j)$ for each pair of communities $C_i$ and $C_j$.
        \State Find the pair of communities $C_i$ and $C_j$ with the maximum $\Delta Q_{ij}$ which satisfies $\Delta Q^P_{ij}<0$.
        \State Merge communities $C_i$ and $C_j$: $C_i=C_i\cap C_j$, and remove $C_j$.
        \State Update the values for each $k\neq i$:
        \begin{gather*}
            e_{ik} \gets e_{ik} + e_{jk}, \quad a_{i} \gets a_i + a_j\\
            e^P_{ik} \gets e^P_{ik} + e^P_{jk}, \quad a^P_{i} \gets a^P_i + a^P_j
        \end{gather*}
    \EndWhile
    \end{algorithmic}
\end{algorithm}
\footnotetext{The code will be provided in GitHub after review.}

The pseudo-code of the proposed FairFN algorithm is given in Algorithm \ref{alg_fair_fn}, which only adds an extra step to optimize the fairness-modularity. Clearly, the FairFN algorithm is initialized by $n$ communities, each of which contains one single vertex; and two communities are chosen to merge in each iteration so that the modularity of $G$ is increased incrementally while the fairness-modularity of $G^P$ is monotonously reduced.
As shown in Algorithm \ref{alg_fair_fn}, $\Delta Q_{ij}$ can be rewritten as a matrix $\Delta Q=\left(\Delta Q_{ij}\right)_{n\times n}$
\begin{equation}
    \Delta Q=2(e-aa^T),
\end{equation}
where $e=(e_{ij})$ and $a=(a_i)$ is a column vector. Updating $e$ and $a$ involves only matrix summation. And a similar operation can also be performed on $\Delta Q^P_{ij}$. 
Therefore, the steps of Algorithm \ref{alg_fair_fn} can be implemented in parallel to achieve a significant speed-up. 
By simple computation, we see that the computational complexity of the proposed FairFN algorithm is $O((m+n)n)$.

In practice, a relaxation on the stopping condition is applied to avoid unbalanced results, such that
\begin{equation}
    \Delta Q_{ij}\le -\frac{\alpha}{2m},
\end{equation}
where the introduced hyperparameter $\alpha \ge 0$. Here, $\alpha$ has an obvious practical implication. When $\alpha+\beta_{ij}$ ($\beta_{ij}$ represents the number of edges between communities $i$ and $j$) falls below the expected value $\frac{a_i a_j}{2m}$ in the random network, the algorithm stops. For analysis of $\alpha$, see Appendix \ref{appendix_threshold}.

\section{Experiments}
\label{section_experiments}

In this section, experiments over the LFR benchmark are applied to validate the effect of the proposed FairFN algorithm by comparisons with the FN algorithm. Extensive experiments are conducted on both synthetic and real-world datasets to show that the FairFN algorithm significantly outperforms other state-of-the-art fair partitioning algorithms.

\subsection{Datasets and Validation Measures}
The following datasets are employed for experiments in this work (details are given in Appendix \ref{appendix_dataset}):

\textbf{LFR benchmark:} the Lancichinetti-Fortunato-Radicchi (LFR) benchmark is a widely used synthetic graph generation model designed to evaluate community detection algorithms. We apply it to generate graphs that exhibit realistic and controllable characteristics, such as power-law degree distributions and community structures with varying sizes communities.

\textbf{Synthetic clustering datasets:} two datasets are generated by Scikit-learn \citep{sklearn}, which has $3000$ samples and $2$ attributes.

\textbf{Real-world datasets:} the datasets of Adult, Bank, Census1990, Creditcard and Diabetic are used in our experiments, which can be found at the UCI Machine Learning Repository\footnote{\url{https://archive.ics.uci.edu/datasets}}.

\subsection{Validation Parameters}
In addition, four measures of fairness ratio (FR) introduced by \citet{FairSC_fair_algorithm} and \citet{Bera}, the average Wasserstein distance between probability vectors (AWD) proposed in \citet{wang2019fairdeepclusteringmultistate}, modularity, and fairness-modularity are used to evaluate the performance of the applied algorithms in the experiments. 

The Fairness ratio (FR) of the community $C_u$ is defined by
\begin{equation}
    FR(C_u)=\min_{w}\left[\frac{r(u,w)}{r(w)},\frac{r(w)}{r(u,w)}\right],
\end{equation}
where $r(u,w)=\frac{|C_u\cap P_w|}{|C_u|}$ and $r(w)=\frac{|P_w|}{n}$. FR of the community partition $\mathbb{C}=\{C_1,C_2,\cdots,C_k\}$ is defined by
\begin{equation}
    FR(C)=\min_{C_u \in \mathbb{C}}\, FR(C_u).
\end{equation}

Clearly, $FR(C) \leq 1$ and a larger value of FR is better.

The average Wasserstein distance between probability vectors (AWD) is defined as
\begin{equation}
    AWD=\frac{\sum_{u}|C_u|\times WD(p_u,p)}{n},
\end{equation}
where $p_u$ is the distribution of the protected group in $C_u$ and $p$ is the distribution of the protected group in $V$.
A smaller value of AWD is better.

\subsection{Experiments over LFR Networks}

In this experiment, we present the performance differences between FN and FairFN over the LFR benchmark networks, which highlights the effect of FairFN and also validates our proposed simple but powerful general framework.

To have intuitive recognition for the distinction between FN and FairFN, we draw communities that are detected by FN and FairFN and protected groups in an LFR network with 200 nodes. Figure \ref{fig_lfr_fn_net} and Figure \ref{fig_lfr_fair_fn_net} show that FairFN still has a fine effect of community partitioning. Only a few vertices are improperly classified due to the trade-off between modularity and fairness. Figure \ref{fig_lfr_fn_pgnet} and Figure \ref{fig_lfr_fair_fn_pgnet} show the fairness of each community. At the bottom and upper right of clusters in Figure \ref{fig_lfr_fn_pgnet}, the red nodes are dominant, i.e., the proportion of the protected groups deviates greatly from $1:1$, which is an unfair community partition. And at each cluster in Figure \ref{fig_lfr_fair_fn_pgnet}, the proportion of red and green nodes is $1:1$, meaning a fair community partition.

\begin{figure}[ht]
    \begin{center}
    \subfloat[communities of FN]{\includegraphics[width=.20\textwidth]{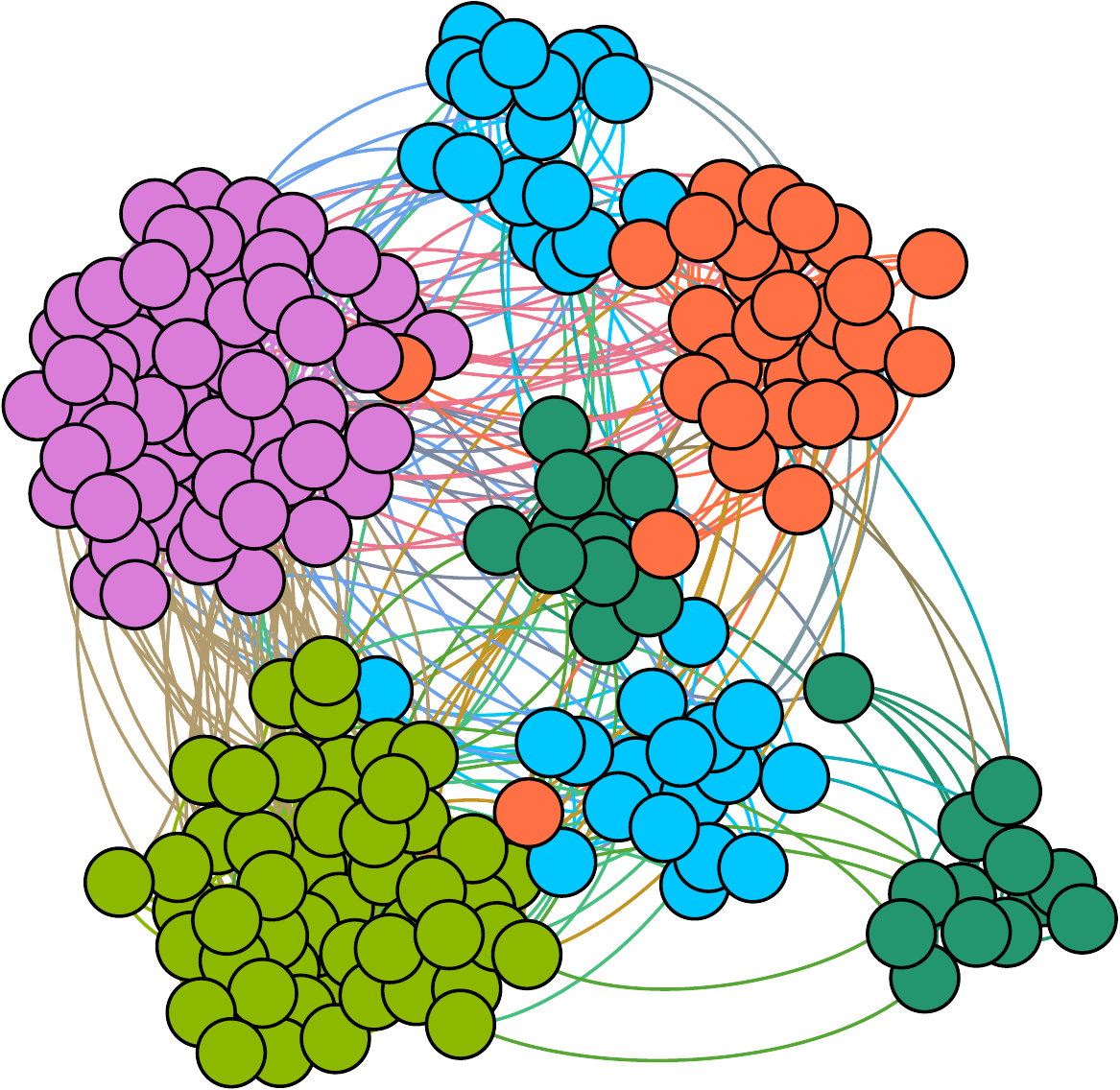}
    \label{fig_lfr_fn_net}}
    \hfill
    \subfloat[communities of FairFN]{\includegraphics[width=.20\textwidth]{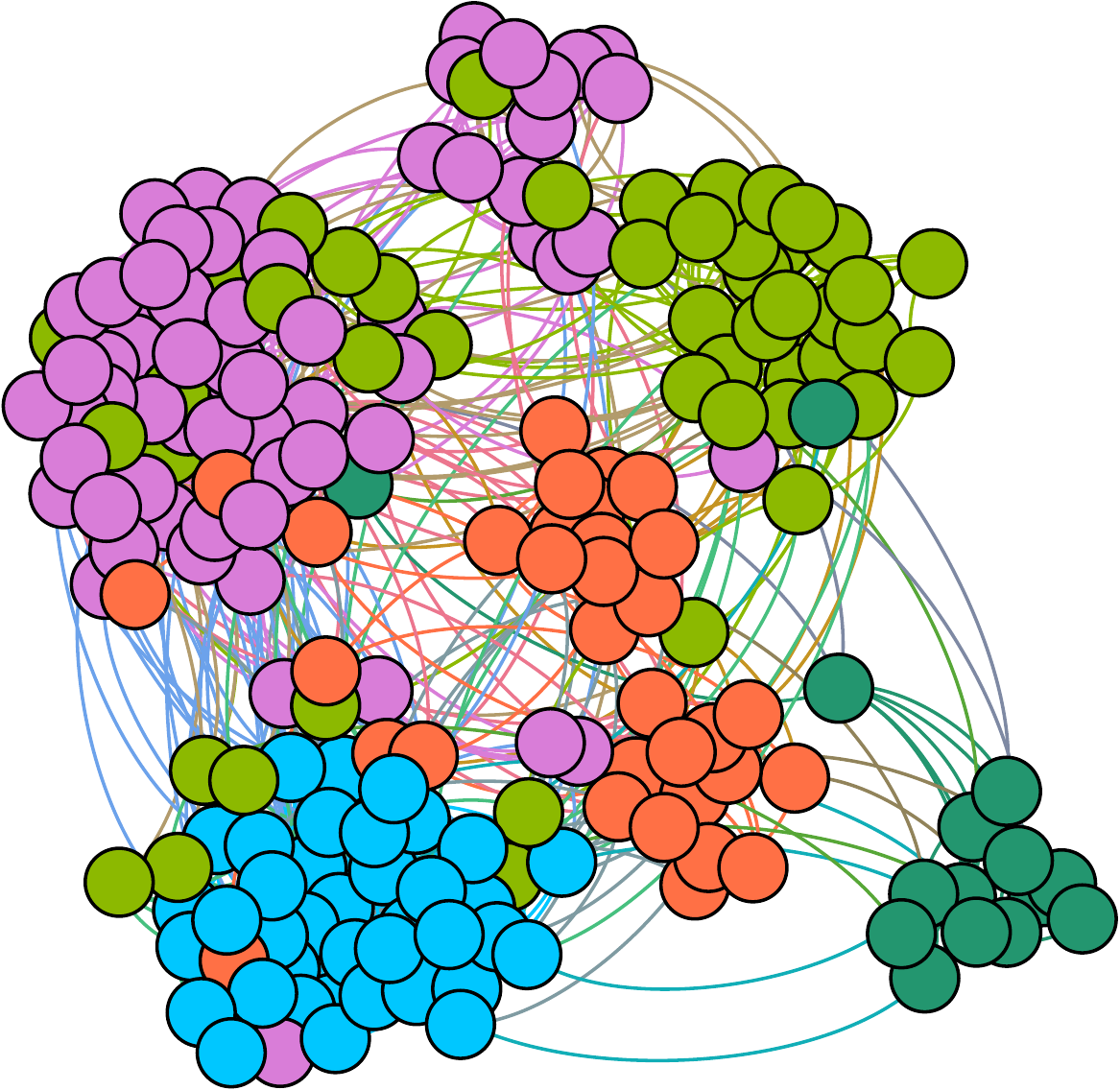}
    \label{fig_lfr_fair_fn_net}}
    \hfill
    \subfloat[protected groups of FN]{\includegraphics[width=.20\textwidth]{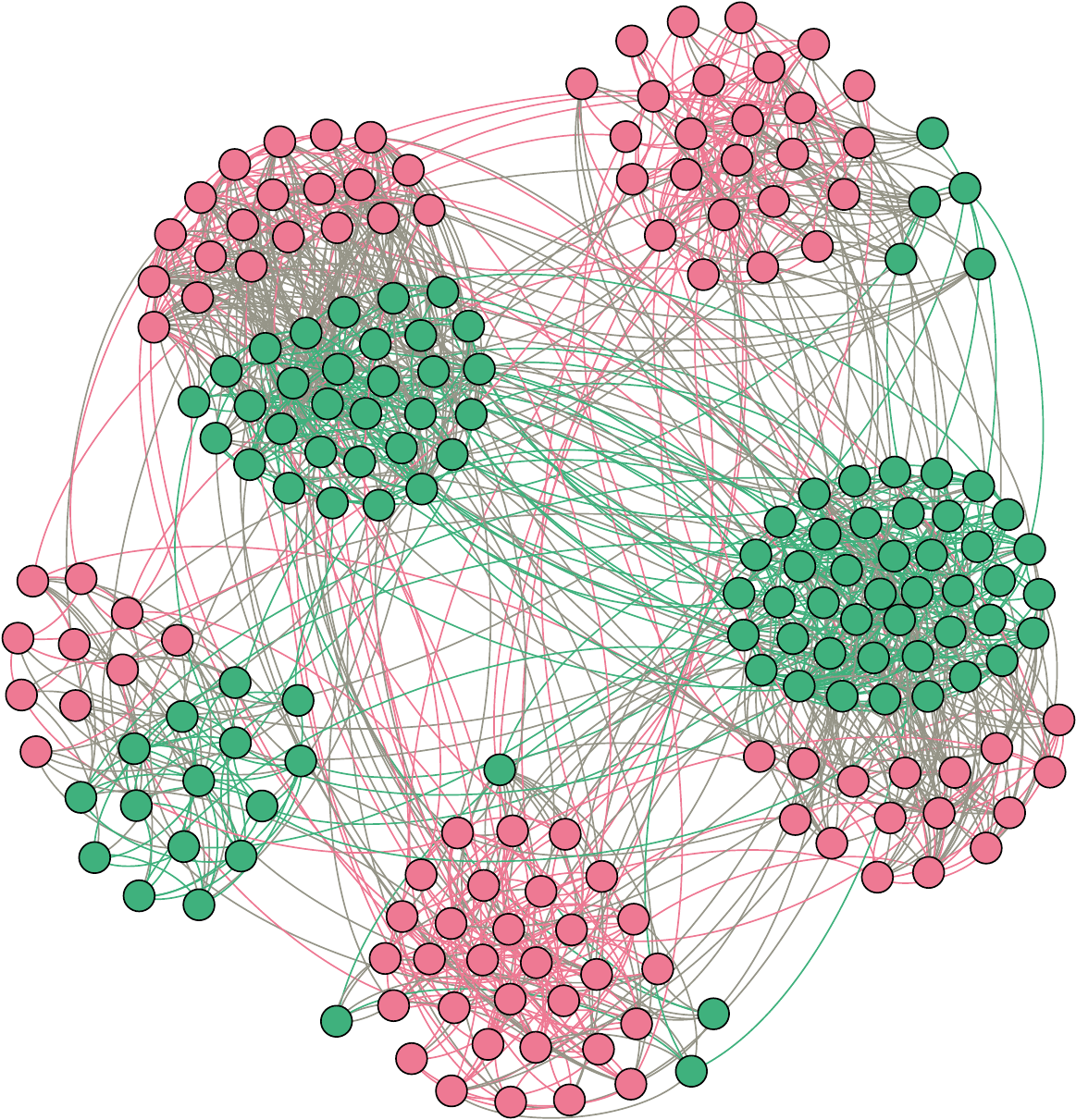}
    \label{fig_lfr_fn_pgnet}}
    \hfill
    \subfloat[protected groups of FairFN]{\includegraphics[width=.20\textwidth]{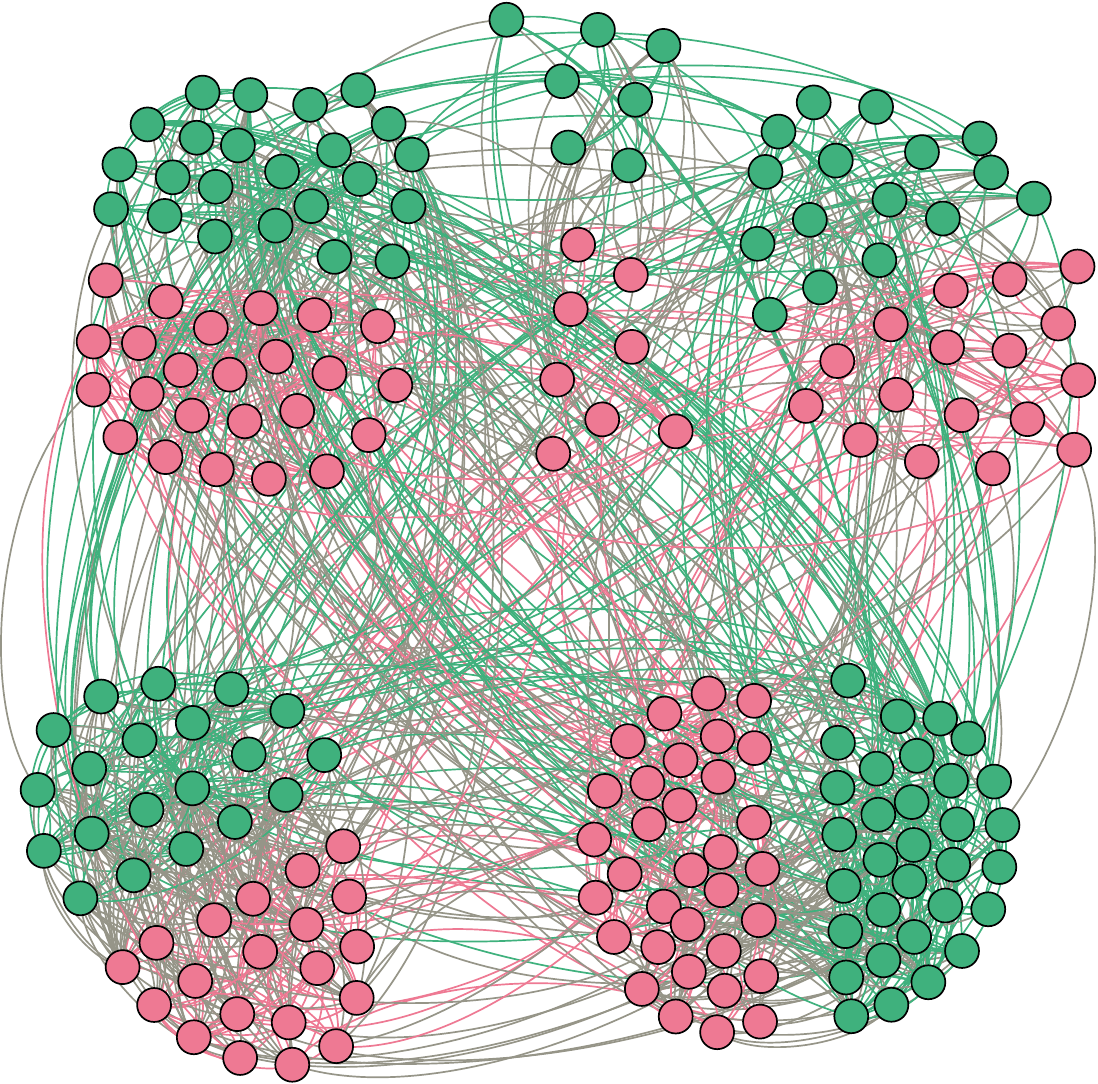}
    \label{fig_lfr_fair_fn_pgnet}}
    \caption{Communities detected and protected groups of each community partition in the LFR network by different algorithms. The network is divided into five communities. Figure \ref{fig_lfr_fn_net} and \ref{fig_lfr_fair_fn_net} show the community partitioning results in different colors. Figure \ref{fig_lfr_fn_pgnet} and \ref{fig_lfr_fair_fn_pgnet} show two protected groups of each community detected by FN or FairFN, and each cluster is a community. Clearly, the proposed FairFN performs much better to achieve community fairness. }
    \label{fig_lfr_net}
    \end{center}
\end{figure}

\begin{figure}[t]
    \begin{center}
    \subfloat[FR]{\includegraphics[width=.3\textwidth]{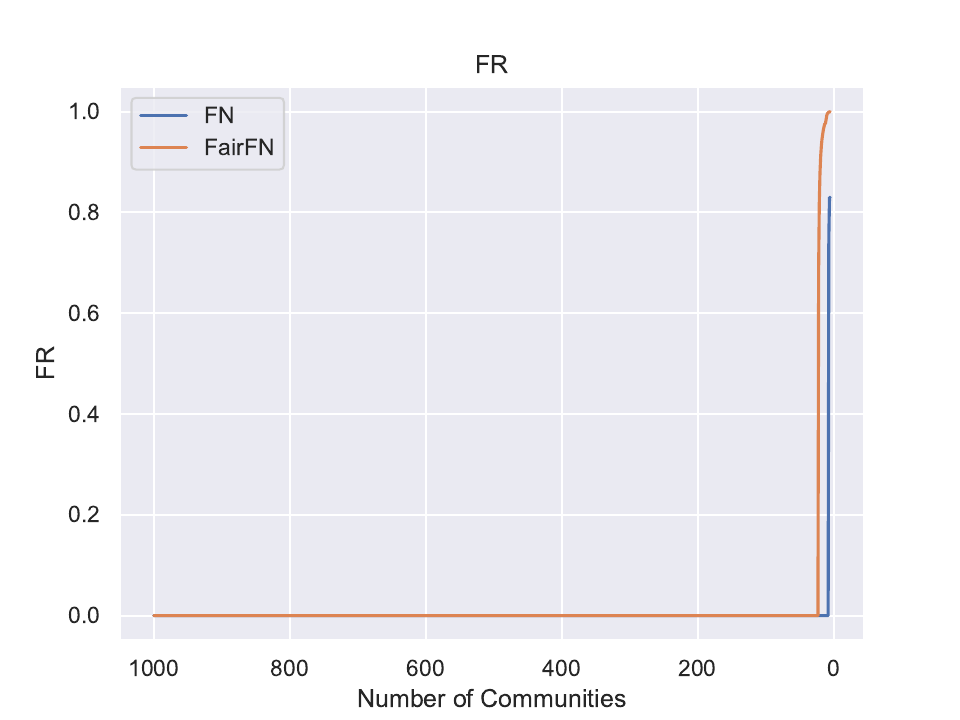}
    \label{fig_lfr_fr}}
    \hfill
    \subfloat[AWD]{\includegraphics[width=.3\textwidth]{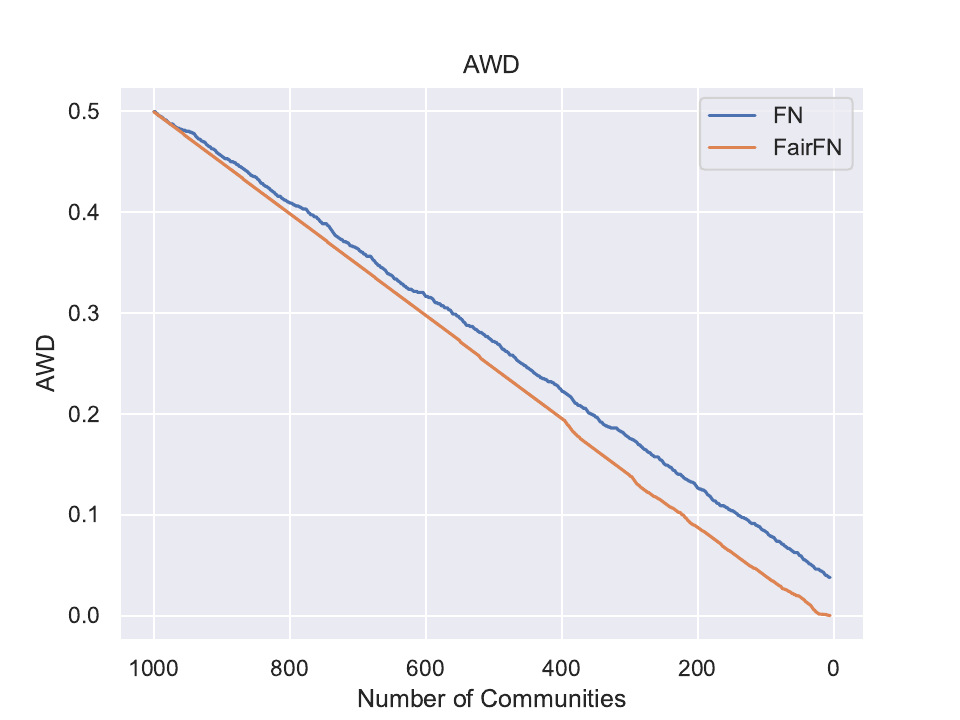}
    \label{fig_lfr_awd}}
    \hfill
    \subfloat[Modularity]{\includegraphics[width=.3\textwidth]{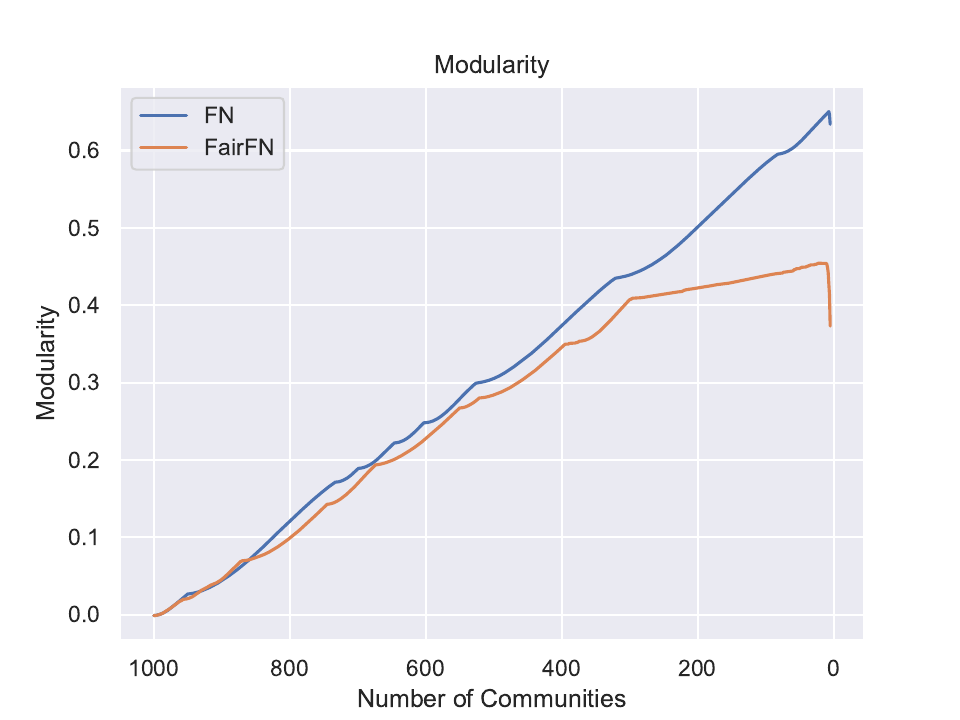}
    \label{fig_lfr_modularity}}
    \quad
    \subfloat[Fairness-Modularity]{\includegraphics[width=.3\textwidth]{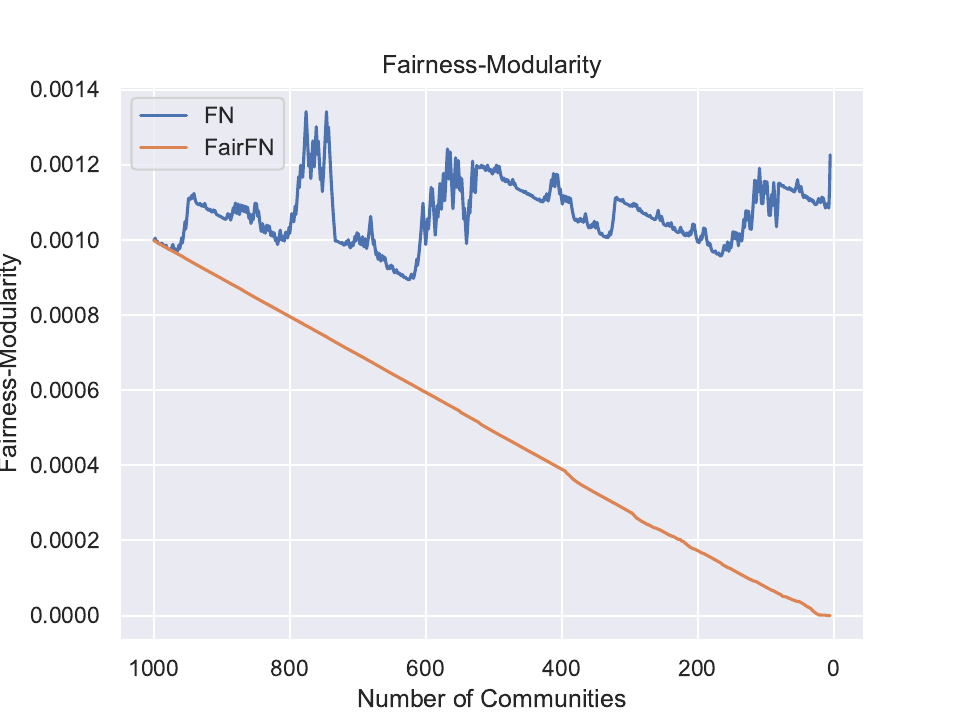}
    \label{fig_lfr_fmodularity}}
    \hfill
    \subfloat[NMI]{\includegraphics[width=.3\textwidth]{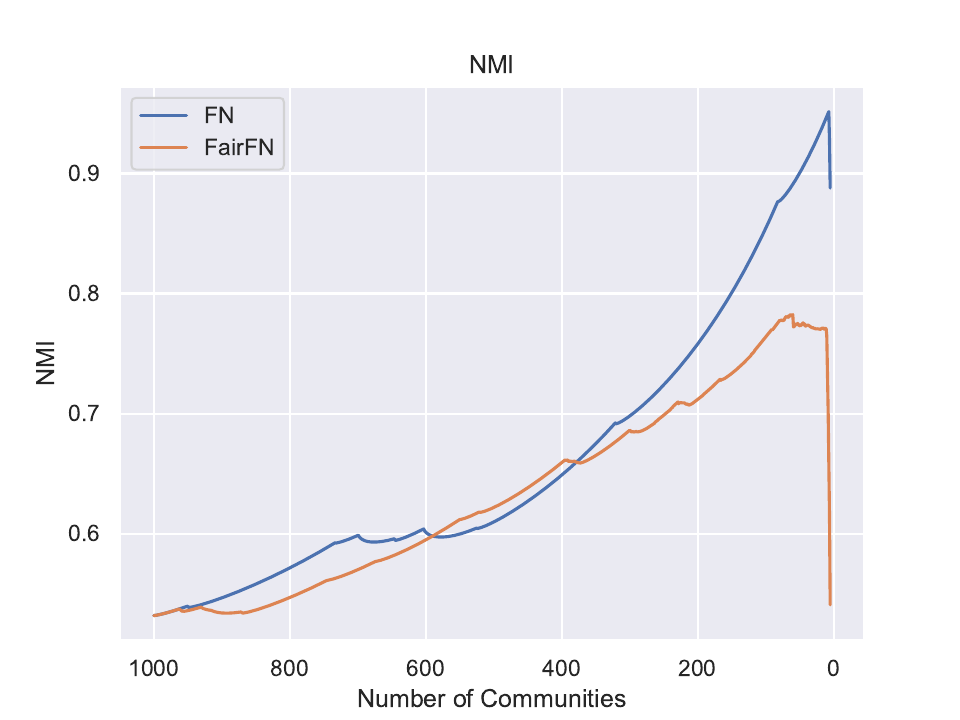} \label{fig_lfr_nmi}}
    \hfill
    \subfloat[thresholds of $\alpha$]{\includegraphics[width=.3\textwidth]{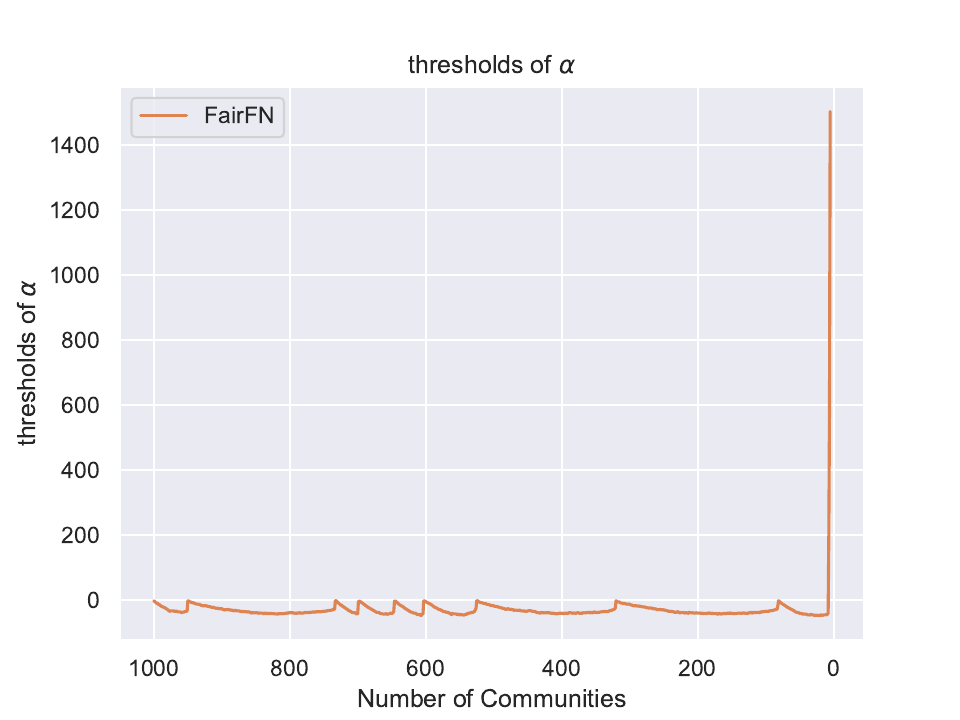} \label{fig_lfr_alpha}}
    \caption{Experiment results in the given LFR network. The number of communities ranges from 5 to 1000 which is decreasing due to the merging process. The thresholds of $\alpha$ are calculated as $2m\Delta Q_{max}$, and $\alpha$ should exceed the thresholds if one wants to allow community merging to continue. The results show that FairFN obtained better fairness with a low impact on modularity.}
    \label{fig_lfr_metric}
    \end{center}
\end{figure}

The changes in the metrics during iteration are shown in Figure \ref{fig_lfr_metric}. From the results, it can easily be seen that FairFN achieves better fairness with a low impact on the effect of community partitioning. What should be noticed is that there is a sharp increase in Figure \ref{fig_lfr_fr}, with the merging process of the communities. This is because the FR is a sensitive metric for the missing protected group in a community. With a greedy strategy, the unfair community will not be considered until all mergers are done that bring greater modularity. Therefore, the value of $\alpha$ must exceed a certain threshold to prevent the disregard of unfair communities. On the other hand, AWD is a more stable metric, and it decreases during merging, indicating that fairness improves steadily throughout the process.

In terms of modularity, as shown in Figure \ref{fig_lfr_modularity}, FairFN and FN exhibit similar growth at the beginning of the merging process, suggesting that the conflict between concentration and fairness is inapparent at the initial stage. The conflict between the two becomes noticeable only toward the end. Additionally, we observe a sharp decrease in modularity at the end of the merging, which is a characteristic of the Fast Newman algorithm. It stems from the growing dominance of the null model’s expected edge term over observed connectivity.
To find a balance between concentration and fairness, we adaptively increase the parameter $\alpha$, which controls the threshold of mergers. Crucially, in late-stage merging, the threshold of $\alpha$ rises rapidly, creating a widening tolerance range for parameter adjustments (Figure \ref{fig_lfr_alpha}). This expanded tolerance correlates with the reduced number of remaining communities: fewer communities inherently allow larger $\alpha$ increments without risking over-merging. The order-of-magnitude differences in adjustable range of $\alpha$ during this phase ensure robustness—by the time communities become sparse, the algorithm prioritizes stability over aggressive merging, effectively mitigating over merging. More results about the threshold of $\alpha$ are given in Appendix \ref{appendix_threshold}.

The results in Figure \ref{fig_lfr_fmodularity} indicate that, compared with the FN algorithm, FairFN can indeed achieve an effective reduction of fairness-modularity.

As for the real communities detection, shown in Figure \ref{fig_lfr_nmi}, the performance of the FairFN algorithm is not as good as that of the FN algorithm. But this is understandable because real communities are often unfairly partitioned. The still high NMI of FairFN algorithm indicates that the communities detected by FairFN don't deviate too much from the original communities, even when fairness is taken into account. Therefore, through a few adjustments, we can transform unfair real communities into fair ones, which holds significant practical value.

\subsection{Comparisons with Baseline Algorithms}
For more details, we compare FairFN with state-of-the-art fair clustering algorithms, including NCut in variational fair clustering (VFC) \citep{VFC_fair_algorithm}, fair
normalized spectral clustering (FSCN) \citep{pmlr-v97-kleindessner19b}, balanced fair k-means clustering (BFKM) \citep{bfkm_algorithm}, and fair clustering ensemble (FCE) \citep{fce_algorithm}. The results of the FN algorithm are also included for comparison. The parameters of algorithms for comparison are listed in Table \ref{table_parameters} (see Appendix \ref{appendix_parameter}).

\begin{table}[!t]
    \caption{Fairness metrics (FR and AWD) for the results in the comparison among algorithms. The upward arrow on the right of the metrics means that the higher the value, the better, and the downward arrow does the opposite. The bold entities indicate the best results.}
    \label{table_fairness}
    \begin{center}
    \begin{tiny}
    \begin{tabular}{ccccccccccccc}
    \toprule
    \multirow{2}{*}{Datasets} & \multicolumn{6}{c}{FR $\uparrow$} & \multicolumn{6}{c}{AWD $\downarrow$}\\
    \cmidrule(r){2-7} \cmidrule(l){8-13} & FN & VFC & FSCN & BFKM & FCE & {\bf FairFN} & FN & VFC & FSCN & BFKM & FCE & {\bf FairFN} \\
    \midrule
    Synthetic\_equal & 0.7606 & 0.8670 & 0.9144 & 0.9667 & 0.9737 & \textbf{0.9855} & 0.0335 & 0.0140 & 0.0097 & 0.0097 & 0.0097 & \textbf{0.0020} \\
    Synthetic\_unequal & 0.0000 & 0.5120 & 0.6984 & 0.6580 & 0.6984 & \textbf{0.7383} & 0.0396 & 0.0208 & 0.0138 & 0.0110 & 0.0125 & \textbf{0.0046} \\     
    Adult & 0.5581 & 0.4381 & 0.7502 & 0.4842 & 0.0000 & \textbf{0.9993} & 0.0963 & 0.0613 & 0.0154 & 0.0567 & 0.0075 & \textbf{0.0001} \\
    Bank & 0.2480 & 0.2722 & 0.7255 & 0.3808 & 0.7653 & \textbf{0.9613} & 0.0776 & 0.1663 & 0.0611 & 0.0844 & 0.0285 & \textbf{0.0025} \\
    Census & 0.0000 & 0.0000 & 0.7697 & 0.6555 & 0.2048 & \textbf{0.9951} & 0.3279 & 0.3098 & 0.0560 & 0.0696 & 0.0894 & \textbf{0.0003} \\
    Creditcard & 0.5229 & 0.7300 & 0.9021 & 0.8621 & 0.8023 & \textbf{0.9836} & 0.0476 & 0.0480 & 0.0234 & 0.0262 & 0.0477 & \textbf{0.0004} \\
    Diabetic & 0.7187 & 0.8665 & 0.8624 & 0.9020 & 0.9473 & \textbf{0.9882} & 0.0256 & 0.0258 & 0.0135 & 0.0159 & 0.0092 & \textbf{0.0006} \\
    \bottomrule
    \end{tabular}
    \end{tiny}
    \end{center}
\end{table}

\begin{table}[!t]
    \caption{Modularity and fairness-modularity for the results in the comparison among algorithms. The legend has the same meaning as the one in Table \ref{table_fairness}. Due to the small value of fairness-modularity, the results of fairness-modularity have been multiplied by a constant $10^2$.}
    \label{table_modularity}
    \begin{center}
    \begin{tiny}
    \begin{tabular}{ccccccccccccc}
    \toprule
    \multirow{2}{*}{Datasets} & \multicolumn{6}{c}{$Q\uparrow$} & \multicolumn{6}{c}{$Q^P\times10^2\downarrow$}\\
    \cmidrule(r){2-7} \cmidrule(l){8-13} & FN & VFC & FSCN & BFKM & FCE & {\bf FairFN} & FN & VFC & FSCN & BFKM & FCE & {\bf FairFN} \\
    \midrule
    Synthetic\_equal & 0.9262 & 0.8238 & 0.7646 & 0.8123 & 0.7725 & \textbf{0.9304} & 0.0318 & 0.0136 & 0.0094 & 0.0068 & 0.0100 & \textbf{0.0001} \\       
    Synthetic\_unequal & \textbf{0.9262} & 0.8136 & 0.7646 & 0.7977 & 0.7667 & 0.9036 & 0.0274 & 0.0161 & 0.0144 & 0.0043 & 0.0103 & \textbf{0.0002} \\     
    Adult & 0.7600 & 0.1207 & 0.1226 & 0.1254 & 0.5804 & \textbf{0.7624} & 0.9441 & 0.2748 & 0.0242 & 0.2251 & 0.0045 & \textbf{0.0000} \\
    Bank & 0.8167 & 0.5351 & 0.4327 & 0.5493 & 0.7404 & \textbf{0.8185} & 0.2036 & 0.9088 & 0.3009 & 0.6919 & 0.0752 & \textbf{0.0001} \\
    Census & \textbf{0.8466} & 0.3129 & 0.4112 & 0.4538 & 0.4624 & 0.8310 & 5.7567 & 3.7494 & 0.5592 & 0.6714 & 1.1200 & \textbf{0.0000} \\
    Creditcard & 0.7178 & 0.4359 & 0.4330 & 0.4519 & 0.6725 & \textbf{0.7222} & 0.2279 & 0.1524 & 0.0502 & 0.0962 & 0.2651 & \textbf{0.0000} \\
    Diabetic & 0.7060 & 0.5434 & 0.2955 & 0.5409 & 0.7119 & \textbf{0.7729} & 0.0930 & 0.0454 & 0.0247 & 0.0318 & 0.0107 & \textbf{0.0000} \\
    \bottomrule
    \end{tabular}
    \end{tiny}
    \end{center}
\end{table}

The comparison results are shown in Table \ref{table_fairness} and Table \ref{table_modularity}. Obviously, FairFN achieves the best fairness and modularity compared to other fair clustering algorithms. It performs well in synthetic clustering datasets and real-world datasets. One interesting finding shows that FairFN is a sensitive algorithm to catch minority groups in the Synthetic-unequal dataset, which is particularly important for practical problems under majority-unbalanced data.
In addition to its fairness performance, FairFN has the advantage of automatically stopping once the appropriate number of communities is reached, without requiring prior knowledge of the number of communities, which is an important feature for real-world community partitioning.

As shown in Table \ref{table_fairness}, we observe that FR of FN outperforms VFC, BFKM and FCE in the Adult dataset, highlighting the neglect of these three algorithms for protected groups in some communities, while they have better AWD. Similarly, in the Bank dataset, FN demonstrates better AWD than BFKM, which can be interpreted in a opposite way. Regarding modularity (see Table \ref{table_modularity}), FairFN also achieves high values, even surpassing FN in some datasets. This suggests that incorporating protected group information can enhance the effect of vanilla greedy community partitioning methods.
Overall, our algorithm achieves excellent results in both fairness and modularity. This success is attributed to our novel fairness-modularity, which brings a fresh perspective to the problem of fair community partitioning.

\section{Discussions and Conclusions}
\label{section_conclusions}

In this paper, we focus on the problem of fair community partitioning and propose a novel fairness-modularity. We prove that minimizing this modularity is equivalent to partitioning communities with fairness, opening up new opportunities for modularity optimization-based community partitioning. Based on this new modularity optimization problem, we present a general solution framework and provide a simple upgrade to the classic FN algorithm, resulting in the efficient FairFN algorithm. Experimental results demonstrate that this simple algorithmic change yields impressive computational outcomes. Compared to baseline algorithms, FairFN significantly outperforms them in both fairness and community structure.

While FairFN demonstrates excellent effect in achieving fair community partitioning, we acknowledge that its time complexity of $O(n(m+n))$ may pose scalability challenges for extremely large networks, which is a common challenge in fair community partitioning. However, this limitation can be solved by further work. Fairness-modularity can be integrated into CNM, Louvain, and many other low-time-complexity algorithms with the guidance of our general framework. In addition, we believe that the FairFN algorithm is already sufficient to demonstrate the significant role of fairness-modularity in fair community partitioning.




Notice that maximum modularity corresponds to an obvious community structure compared with the random network and high aggregation in communities. The minimum modularity corresponds to a random partition and fairness, in another way, homogeneity. It provides a perspective of probability to consider fairness.



\bibliographystyle{icml2025}
\bibliography{reference.bib}

\begin{thebibliography}{25}
\providecommand{\natexlab}[1]{#1}
\providecommand{\url}[1]{\texttt{#1}}
\expandafter\ifx\csname urlstyle\endcsname\relax
  \providecommand{\doi}[1]{doi: #1}\else
  \providecommand{\doi}{doi: \begingroup \urlstyle{rm}\Url}\fi

\bibitem[Bera et~al.(2019)Bera, Chakrabarty, Flores, and Negahbani]{Bera}
Bera, S.~K., Chakrabarty, D., Flores, N.~J., and Negahbani, M.
\newblock \emph{Fair algorithms for clustering}.
\newblock Curran Associates Inc., Red Hook, NY, USA, 2019.

\bibitem[Blondel et~al.(2008)Blondel, Guillaume, Lambiotte, and
  Lefebvre]{FUA_algorithm}
Blondel, V.~D., Guillaume, J.-L., Lambiotte, R., and Lefebvre, E.
\newblock Fast unfolding of communities in large networks.
\newblock \emph{Journal of Statistical Mechanics: Theory and Experiment},
  2008\penalty0 (10):\penalty0 P10008, oct 2008.

\bibitem[Brandes et~al.(2008)Brandes, Delling, Gaertler, Gorke, Hoefer,
  Nikoloski, and Wagner]{np_completeness}
Brandes, U., Delling, D., Gaertler, M., Gorke, R., Hoefer, M., Nikoloski, Z.,
  and Wagner, D.
\newblock On modularity clustering.
\newblock \emph{IEEE Transactions on Knowledge and Data Engineering},
  20\penalty0 (2):\penalty0 172--188, 2008.

\bibitem[Chierichetti et~al.(2017)Chierichetti, Kumar, Lattanzi, and
  Vassilvitskii]{FairSC_fair_algorithm}
Chierichetti, F., Kumar, R., Lattanzi, S., and Vassilvitskii, S.
\newblock Fair clustering through fairlets.
\newblock In \emph{Proceedings of the 31st International Conference on Neural
  Information Processing Systems}, NIPS'17, pp.\  5036–5044, Red Hook, NY,
  USA, 2017. Curran Associates Inc.
\newblock ISBN 9781510860964.

\bibitem[Clauset et~al.(2004)Clauset, Newman, and Moore]{cnm_algorithm}
Clauset, A., Newman, M. E.~J., and Moore, C.
\newblock Finding community structure in very large networks.
\newblock \emph{Phys. Rev. E}, 70:\penalty0 066111, Dec 2004.

\bibitem[Di et~al.(2012)Di, Bo, Jie, Da-you, and Dong-xiao]{RWACO_algorithm}
Di, J., Bo, Y., Jie, L., Da-you, L., and Dong-xiao, H.
\newblock Ant colony optimization based on random walk for community detection
  in complex networks.
\newblock \emph{Journal of Software}, 23\penalty0 (3):\penalty0 451, 2012.

\bibitem[Fortunato(2010)]{fortunato2010community}
Fortunato, S.
\newblock Community detection in graphs.
\newblock \emph{Physics Reports}, 486\penalty0 (3-5):\penalty0 75--174, 2010.

\bibitem[Fortunato \& Newman(2022)Fortunato and Newman]{fortunato202220}
Fortunato, S. and Newman, M.~E.
\newblock 20 years of network community detection.
\newblock \emph{Nature Physics}, 18\penalty0 (8):\penalty0 848--850, 2022.

\bibitem[Jin et~al.(2011)Jin, Liu, Yang, Liu, He, and Tian]{FNCA_algorithm}
Jin, D., Liu, D.-Y., Yang, B., Liu, J., He, D.-X., and Tian, Y.
\newblock Fast complex network clustering algorithm using local detection.
\newblock \emph{Dianzi Xuebao(Acta Electronica Sinica)}, 39\penalty0
  (11):\penalty0 2540--2546, 2011.

\bibitem[Khademi \& Honavar(2020)Khademi and Honavar]{COMPAS}
Khademi, A. and Honavar, V.
\newblock Algorithmic bias in recidivism prediction: A causal perspective
  (student abstract).
\newblock \emph{Proceedings of the AAAI Conference on Artificial Intelligence},
  34\penalty0 (10):\penalty0 13839--13840, Apr. 2020.

\bibitem[Kleindessner et~al.(2019)Kleindessner, Samadi, Awasthi, and
  Morgenstern]{pmlr-v97-kleindessner19b}
Kleindessner, M., Samadi, S., Awasthi, P., and Morgenstern, J.
\newblock Guarantees for spectral clustering with fairness constraints.
\newblock In Chaudhuri, K. and Salakhutdinov, R. (eds.), \emph{Proceedings of
  the 36th International Conference on Machine Learning}, volume~97 of
  \emph{Proceedings of Machine Learning Research}, pp.\  3458--3467. PMLR,
  09--15 Jun 2019.

\bibitem[Leicht \& Newman(2008)Leicht and Newman]{directed_modularity}
Leicht, E.~A. and Newman, M. E.~J.
\newblock Community structure in directed networks.
\newblock \emph{Phys. Rev. Lett.}, 100:\penalty0 118703, Mar 2008.

\bibitem[Li et~al.(2024)Li, Hu, Du, Chen, Jiang, and Zhou]{multi_view}
Li, R., Hu, H., Du, L., Chen, J., Jiang, B., and Zhou, P.
\newblock One-stage fair multi-view spectral clustering.
\newblock In \emph{Proceedings of the 32nd ACM International Conference on
  Multimedia}, MM '24, pp.\  1407–1416, New York, NY, USA, 2024. Association
  for Computing Machinery.
\newblock ISBN 9798400706868.

\bibitem[Manolis \& Pitoura(2024)Manolis and Pitoura]{10.1145/3625007.3627518}
Manolis, K. and Pitoura, E.
\newblock Modularity-based fairness in community detection.
\newblock In \emph{Proceedings of the 2023 IEEE/ACM International Conference on
  Advances in Social Networks Analysis and Mining}, ASONAM '23, pp.\
  126–130, New York, NY, USA, 2024. Association for Computing Machinery.
\newblock ISBN 9798400704093.

\bibitem[Newman(2006)]{newman2006modularity}
Newman, M.~E.
\newblock Modularity and community structure in networks.
\newblock \emph{Proceedings of the National Academy of Sciences}, 103\penalty0
  (23):\penalty0 8577--8582, 2006.

\bibitem[Newman(2004)]{fn_algorithm}
Newman, M. E.~J.
\newblock Fast algorithm for detecting community structure in networks.
\newblock \emph{Phys. Rev. E}, 69:\penalty0 066133, Jun 2004.

\bibitem[Newman \& Girvan(2004)Newman and Girvan]{gn_algorithm}
Newman, M. E.~J. and Girvan, M.
\newblock Finding and evaluating community structure in networks.
\newblock \emph{Phys. Rev. E}, 69:\penalty0 026113, Feb 2004.

\bibitem[Pan et~al.(2024)Pan, Zhong, and Qian]{bfkm_algorithm}
Pan, R., Zhong, C., and Qian, J.
\newblock Balanced fair k-means clustering.
\newblock \emph{IEEE Transactions on Industrial Informatics}, 20\penalty0
  (4):\penalty0 5914--5923, 2024.

\bibitem[Pedregosa et~al.(2011)Pedregosa, Varoquaux, Gramfort, Michel, Thirion,
  Grisel, Blondel, Prettenhofer, Weiss, Dubourg, Vanderplas, Passos,
  Cournapeau, Brucher, Perrot, and Duchesnay]{sklearn}
Pedregosa, F., Varoquaux, G., Gramfort, A., Michel, V., Thirion, B., Grisel,
  O., Blondel, M., Prettenhofer, P., Weiss, R., Dubourg, V., Vanderplas, J.,
  Passos, A., Cournapeau, D., Brucher, M., Perrot, M., and Duchesnay, E.
\newblock Scikit-learn: Machine learning in python.
\newblock \emph{J. Mach. Learn. Res.}, 12\penalty0 (null):\penalty0
  2825–2830, November 2011.
\newblock ISSN 1532-4435.

\bibitem[Raghavan et~al.(2007)Raghavan, Albert, and Kumara]{lpa_algorithm}
Raghavan, U.~N., Albert, R., and Kumara, S.
\newblock Near linear time algorithm to detect community structures in
  large-scale networks.
\newblock \emph{Phys. Rev. E}, 76:\penalty0 036106, Sep 2007.

\bibitem[Rosvall \& Bergstrom(2007)Rosvall and Bergstrom]{encoding_detection}
Rosvall, M. and Bergstrom, C.~T.
\newblock An information-theoretic framework for resolving community structure
  in complex networks.
\newblock \emph{Proceedings of the National Academy of Sciences}, 104\penalty0
  (18):\penalty0 7327--7331, 2007.

\bibitem[Wang \& Davidson(2019)Wang and
  Davidson]{wang2019fairdeepclusteringmultistate}
Wang, B. and Davidson, I.
\newblock Towards fair deep clustering with multi-state protected variables,
  2019.

\bibitem[Wang et~al.(2023)Wang, Lu, Davidson, and Bai]{pmlr-v206-wang23h}
Wang, J., Lu, D., Davidson, I., and Bai, Z.
\newblock Scalable spectral clustering with group fairness constraints.
\newblock In Ruiz, F., Dy, J., and van~de Meent, J.-W. (eds.),
  \emph{Proceedings of The 26th International Conference on Artificial
  Intelligence and Statistics}, volume 206 of \emph{Proceedings of Machine
  Learning Research}, pp.\  6613--6629. PMLR, 25--27 Apr 2023.

\bibitem[Zhou et~al.(2025)Zhou, Li, Ling, Du, and Liu]{fce_algorithm}
Zhou, P., Li, R., Ling, Z., Du, L., and Liu, X.
\newblock Fair clustering ensemble with equal cluster capacity.
\newblock \emph{IEEE Transactions on Pattern Analysis and Machine
  Intelligence}, 47\penalty0 (3):\penalty0 1729--1746, 2025.

\bibitem[Ziko et~al.(2021)Ziko, Yuan, Granger, and Ayed]{VFC_fair_algorithm}
Ziko, I.~M., Yuan, J., Granger, E., and Ayed, I.~B.
\newblock Variational fair clustering.
\newblock In \emph{AAAI Conference on Artificial Intelligence}, 2021.

\end{thebibliography}


\newpage
\appendix

\section{Proof of Lemma \ref{thm_corr}}
\label{appendix_lemma_proof}
\begin{proof}
    It is obvious that
    $m^{directed}=\sum_i{k^{in}_i}+\sum_i{k^{out}_i}=2m$. By the modularity definition \eqref{eq_modularity_directed} of any directed graph $G_d$, we prove that
    \begin{align}
        Q&=\frac{1}{m^{directed}}\sum_{ij}\sum_{u} \left[A_{ij}-\frac{k^{in}_ik^{out}_j}{m^{directed}}\right]S_{iu}S_{ju}\\
        &=\frac{1}{2m}\sum_{ij}\sum_{u} \left[A_{ij}-\frac{k_ik_j}{2m}\right]S_{iu}S_{ju}.
    \end{align}
\end{proof}

\section{Proof of Theorem \ref{thm-fm}}
\label{appendix_theorem_proof}
\begin{proof}
    Considering $\Tr\left(S^TDS\right)$, it can be expanded into a summation
    \begin{align}
        \Tr\left(S^TDS\right)&=\sum_{ij} \sum_{u}D_{ij}\delta_{C_u}(i,j)\\
        &=\sum_{ij}\sum_{u}\sum_{w}\delta_{C_u}(i,j)\delta_{P_w}(i,j)\\
        &=\sum_{u}\sum_{w}\left(\sum_{v_i\in C_u\cap P_w}1\right)^2\\
        &=\sum_{u=1}^k\sum_{w=1}^r |C_u\cap P_w|^2,
    \end{align}
    where $\delta_{P_w}(i,j)=1$ if $v_i,v_j$ are in the same $P_w$, and $\delta_{P_w}(i,j)=0$ otherwise. Similarly, $\delta_{C_u}(i,j)=1$ if $v_i,v_j$ are in the same $C_u$, and $\delta_{C_u}(i,j)=0$ otherwise.
    
    Then considering
    \begin{equation}
        \Tr\left(S^TK^P(K^P)^TS\right)=\|(K^P)^TS\|_F^2,
    \end{equation}
    $i$-th entry of $(K^P)^TS$ is $\sum_{w=1}^r|C_u\cap P_w||P_w|$. Thus
    \begin{equation}
        \|(K^P)^TS\|_F^2=\sum_{u=1}^k\left(\sum_{w=1}^r|C_u\cap P_w||P_w|\right)^2.
    \end{equation}
    
    Using Cauchy-Schwarz Inequality, we have
    \begin{equation}\label{cauchy_schwarz}
        \left(\sum_{w=1}^r|C_u\cap P_w||P_w|\right)^2\le \sum_{w=1}^r|C_u\cap P_w|^2\sum_{w=1}^r|P_w|^2.
    \end{equation}
    Thus
    \begin{equation}
    \begin{split}
        Q^P&=\frac{1}{2m^P}\left[\sum_{u=1}^k\sum_{w=1}^r |C_u\cap P_w|^2-\frac{1}{2m^P}\sum_{u=1}^k\left(\sum_{w=1}^r|C_u\cap P_w||P_w|\right)^2\right]\\
        &\ge \frac{1}{2m^P}\left[\sum_{u=1}^k\sum_{w=1}^r |C_u\cap P_w|^2-\frac{1}{2m^P}\sum_{w=1}^r|C_u\cap P_w|^2\sum_{w=1}^r|P_w|^2\right]\\
        &=0.
    \end{split}
    \label{eq_fm_eq0}
    \end{equation}
    
    The equality of \eqref{cauchy_schwarz} and \eqref{eq_fm_eq0} holds when
    \begin{equation} \label{eq_equality}
        \frac{|C_u\cap P_1|}{|P_1|}=\frac{|C_u\cap P_2|}{|P_2|}=\cdots=\frac{|C_u\cap P_r|}{|P_r|}=c
    \end{equation}
    for any $u=1,\cdots,k$.
    
    Observing that $\sum_{w=1}^r|C_u\cap P_w|=|C_u|$ and $\sum_{w=1}^r|P_w|=n$, we obtain
    \begin{equation}
        c=\frac{|C_u|}{n}.
    \end{equation}
    So
    \begin{equation}
        \frac{|C_u\cap P_w|}{|C_u|}=\frac{|P_w|}{n}
    \end{equation}
    for any $u$ and $w$, which means that the partitioning is fair.
\end{proof}

\section{Properties of fairness-modularity}
\label{appendix_property}

The range of the {\em fairness-modularity} $Q^P$ can be further obtained as follows:
\begin{corollary}
    $0\le Q^P\le 1-\frac{1}{2m^P}$.
    \label{thm_range}
\end{corollary}
\begin{proof}
    According to the equation \eqref{eq_fm_eq0}, the lower bound of fairness-modularity is 0. The upper bound of fairness-modularity is similar to the clique \citep{np_completeness}. Suppose that the protected group network is divided into $r$ communities and each maximal clique is a community (in other words, it means $\mathbb{C}=\mathbb{P}$). This partition has modularity
    \begin{align}
        &\quad\frac{1}{2m^P}\left[\sum_{u=1}^k\sum_{w=1}^r |P_u\cap P_w|^2-\frac{1}{2m^P}\sum_{u=1}^k\left(\sum_{w=1}^r|P_u\cap P_w||P_w|\right)^2\right]\\
        &=\frac{1}{2m^P}\left[\sum_{w=1}^r|P_w|^2-\frac{1}{2m^P}\sum_{w=1}^r|P_w|^2\right]\\
        &=1-\frac{1}{2m^P}.
    \end{align}
\end{proof}

We can also prove that the fairness-modularity in the situation where every vertex belongs to a different community is
\begin{equation}
    \frac{1}{2m^P}\left(n-\frac{1}{2m^P}\sum_{w=1}^r|P_w|^3\right).
\end{equation}
This shows that the situation in which each node belongs to a different community is not a solution of $Q^P=0$. It excludes some trivial cases.

From the equality \eqref{eq_equality}, it is easy to see that
\begin{proposition}
    Given the fair community partition $\mathbb{C}=\{C_1,C_2,\cdots,C_k\}$, if the two communities $C_i$ and $C_j$ are merged into one community $\widetilde{C} = C_i\cup C_j$, then the new community partition
    \[
    \mathbb{C}=(\{C_1,C_2,\cdots,C_k\}\backslash\{C_i,C_j\})\cup\{\widetilde{C}\}
    \]
    is also a fair community partition.
\end{proposition}
This result implies that the number of communities does not affect the minimum of $Q^P$. Thus, the best number of communities depends only on the modularity $Q$.

\section{Datasets details}
\label{appendix_dataset}

\textbf{LFR benchmark:} The network generated by the LFR benchmark has the following parameters: 1000 nodes, an average degree exponent of 2, a community size exponent of 1.1, and a mixing parameter of 0.1, which features a minimum degree of 20 and a maximum degree of 100. The protected groups of vertices in the LFR network of Figure \ref{fig_lfr_prop} and Figure \ref{fig_lfr_metric} are uniformly sampled from the distribution $(0.5,0.5)$.

For the protected groups of Figure \ref{fig_lfr_net}, we randomly assign a probability $p_k$ ranging from 0.2 to 0.8 for each community $k$ and generate protected groups in the community $k$ from a binomial distribution $B(1,p_k)$. The method can generate unfair communities, while the proportion of protected groups in the whole network keeps nearly $1:1$.

\textbf{Synthetic clustering datasets:} The protected groups are obtained by random sampling of a given distribution. Synthetic-equal dataset has a distribution of $(0.5,0.5)$, which is a balanced dataset. Synthetic-unequal dataset has a distribution of $(0.02,0.48,0.5)$, which is an unbalanced dataset.

\textbf{Real datasets:} Adult dataset contains 32561 samples. Five numeric features of the dataset are chosen for using, and the protected group is gender.
Bank dataset contains 41188 samples. Six numeric attributes are chosen, and the protected group is marital status. We remove the samples with an "unknown" marital status.
Census1990 dataset is a large dataset that contains 2458285 samples. Twenty-five numeric attributes is chosen, and the protected group is gender.
Creditcard dataset contains 30000 samples. Attributes LIMIT\_BAL, AGE, BILL\_AMT1$\sim$6 and PAY\_AMT1$\sim$6 are chosen, and the protected group is sex.
Diabetic dataset contains 101766 samples. Seven numeric attributes are chosen, and the protected group is gender.
All of them are subsampled to 5000 samples.

FairFN processes network data, and we construct the network using $k$-nearest neighbor ($k=10$), as in VFC and FSCN.

\section{Parameters of algorithms}
\label{appendix_parameter}

The parameters of algorithms used in the comparison are listed in Table \ref{table_parameters}. Specifically, in FCE, $\lambda_2$ is constant 0, the number of clusters is 5, and the ensemble of labels is 10.

All experiments have been done on an NVIDIA RTX 3050 Laptop. The running time of the FairFN algorithm in each dataset is less than 2 minutes.

\begin{table}[ht]
    \caption{Parameters of algorithms.}
    \label{table_parameters}
    \begin{center}
    \begin{small}
    \begin{tabular}{ccccccc}
        \toprule
        \multirow{2}{*}{Datasets} & \multicolumn{2}{c}{VFC} & FSCN & BFKM & FCE & FairFN\\
        \cmidrule{2-7} & k & $\lambda$ & k & k & $\lambda_1$ & $\alpha$\\
        \midrule
        Synthetic-equal & 10 & 10 & 5 & 10 & 0.001 & 4\\
        Synthetic-unequal & 10 & 10 & 5 & 10 & 0.001 & 64\\
        Adult & 10 & 9000 & 5 & 10 & 1 & 8\\
        Bank & 10 & 9000 & 5 & 6 & 1 & 4\\
        Census1990 & 20 & 500000 & 5 & 5 & 1 & 4\\
        Creditcard & 10 & 9000 & 5 & 5 & 0.001 & 4\\
        Diabetic & 10 & 9000 & 5 & 5 & 0.001 & 4\\
        \bottomrule
    \end{tabular}
    \end{small}
    \end{center}
\end{table}

\section{More results in the LFR network}
Further results are presented for the LFR network with 1000 nodes. Figure \ref{fig_lfr_prop} shows the proportion of protected groups in different communities. In Figure \ref{fig_lfr_fn_prop}, the difference in the proportion of the protected groups between a community and the whole network is up to 0.15 while the one in Figure \ref{fig_lfr_fair_fn_prop} is only up to 0.005. This marked difference highlights the improved fairness achieved by FairFN compared to FN.

\begin{figure}[ht]
    \begin{center}
    \subfloat[FN]{\includegraphics[width=.48\textwidth]{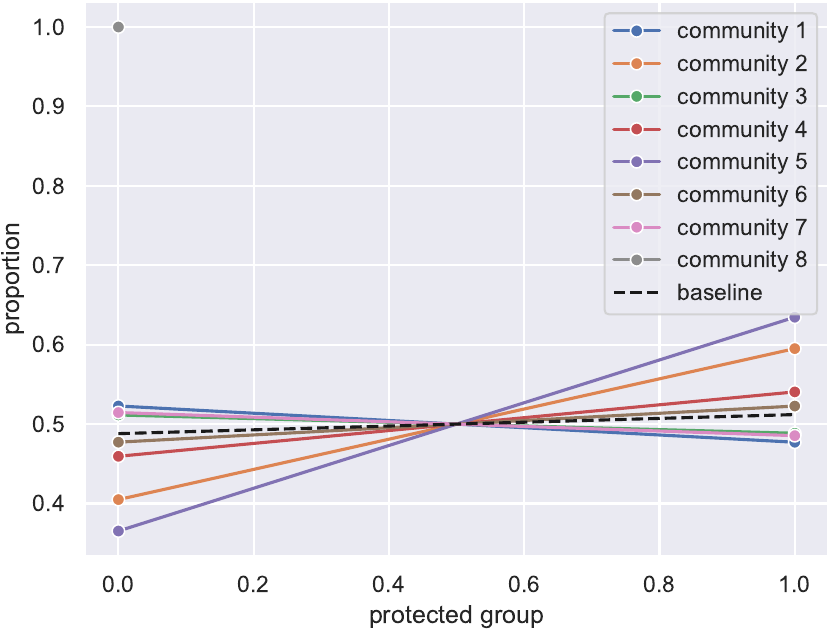}
    \label{fig_lfr_fn_prop}}
    \hfill
    \subfloat[FairFN]{\includegraphics[width=.48\textwidth]{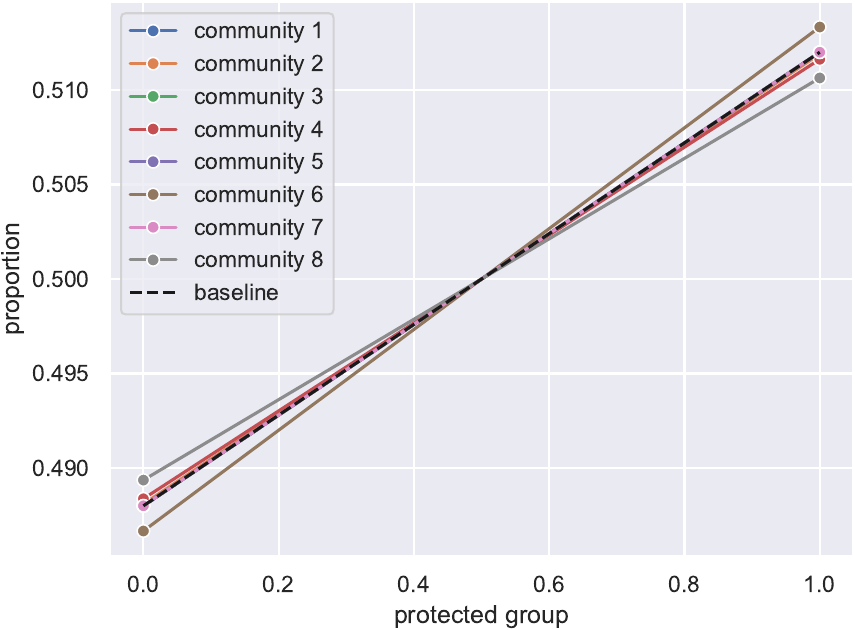}
    \label{fig_lfr_fair_fn_prop}}
    \caption{The proportion of protected groups in different communities. The X-axis represents different protected groups in a community. The results show that the proportion of the protected group in communities detected by FairFN is nearly consistent with one in the whole network, which indicates that the fairness of FairFN is better than FN.}
    \label{fig_lfr_prop}
    \end{center}
\end{figure}

\section{More experiments in unbalanced datasets and weighted networks}
\subsection{Unbalanced datasets}
To further validate the robustness of FairFN, we conducted additional experiments on even more unbalanced datasets. Similar to synthetic datasets, every unbalanced dataset contains 3000 samples, and contains at least one type of protected group with a small proportion. The protected groups are randomly sampled from given distributions. The $\alpha$ parameter of FairFN is set to 64.

The results are shown in Table \ref{table_unbalanced}. Notably, in the $(0.01, 0.49, 0.5)$ dataset, where the group sizes are $(36, 1454, 1510)$ nodes, FairFN identifies 17 communities, demonstrating that the minority group (0.01) is evenly distributed rather than being isolated or marginalized. These results confirm that FairFN maintains fairness and performs well even in highly unbalanced scenarios. FairFN demonstrates stable performance on unbalanced datasets. The FR value decreases on extremely unbalanced datasets, but this is because the FR varies greatly when the number of vertices is small. A stable AWD value is more indicative of the robustness of FairFN, which indicates that our algorithm is well-suited for handling unbalanced data.

\begin{table}[htbp]
    \caption{The results of FairFN on unbalanced datasets. The results demonstrate that FairFN has excellent robustness on unbalanced datasets.}
    \label{table_unbalanced}
    \begin{center}
    \begin{normalsize}
    \begin{tabular}{ccccc}
    \toprule
    Proportion of protected groups & FR & AWD & $Q$ & $Q^P\times 10^6$\\
    \midrule
    $(0.02,0.05,0.93)$ & 0.7996 & 0.0024 & 0.8841 & 0.2906\\
    $(0.05,0.45,0.5)$ & 0.8203 & 0.0057 & 0.8977 & 3.2764\\
    $(0.02,0.48,0.5)$ & 0.7383 & 0.0046 & 0.9036 & 1.9500\\
    $(0.01,0.49,0.5)$ & 0.4845 & 0.0048 & 0.9114 & 2.0704\\
    \bottomrule
    \end{tabular}
    \end{normalsize}
    \end{center}
\end{table}

\subsection{Weighted networks}
In the general framework, we impose a strong restriction on fairness-modularity, that is, only communities with $\Delta Q^P < 0$ are merged. This ensures that, regardless of how the weights of the edges in the observed network change, the algorithm can achieve good fairness. To illustrate this statement, we tested FairFN in the weighted observed network. The experiment was conducted on an LFR network with 1000 nodes and 6 ground-truth communities. The very high intra-community edge weights (50) and inter-community edge weights (10) are assigned in the observed network. The protected groups of vertices were randomly sampled from the distribution $(0.5, 0.5)$. In this experiment, the total edge weight of the observed network far exceeded that of the protected group network. An unweighted observed network was included as a control. The network structure of the unweighted observed network is the same as that of the weighted network, except that all the edge weights are set to 1. The $\alpha$ parameter of FairFN was set to 8 in all cases.

The results are shown in Table \ref{table_weighted}, where FN is included as the control group for comparison. FairFN achieved better performance than FN in terms of fairness, that is, a higher FR, and lower AWD and $Q^P$, indicating that FairFN can still achieve fair community partitioning on weighted networks. Moreover, the fairness performance of the FairFN algorithm is similar both on the weighted observed network and the unweighted observed network, which shows that the total edge weight of the observed network does not affect the fairness of FairFN.

\begin{table}[htbp]
    \caption{The results of FN and FairFN on weighted and unweighted LFR networks. The results of our FairFN algorithm are better than FN. Moreover, regardless of whether the observed network is weighted or not, the fairness performance of FairFN is not significantly affected.}
    \label{table_weighted}
    \begin{center}
    \begin{tabular}{cccccc}
        \toprule
        Algorithm & observed network & FR & AWD & $Q$ & $Q^P\times 10^2$ \\
        \midrule
        FN & weighted & 0.8023 & 0.0227 & 0.7348 & 0.0678 \\
        FairFN & weighted & 0.9860 & 0.0014 & 0.7080 & 0.0002 \\
        FN & unweighted & 0.8162 & 0.0217 & 0.6189 & 0.0612 \\
        FairFN & unweighted & 0.9951 & 0.0009 & 0.5950 & 0.0001 \\
        \bottomrule
    \end{tabular}
\end{center}
\end{table}

Based on the performance of FairFN on the extremely unbalanced dataset (Table \ref{table_unbalanced}) and the weighted network (Table \ref{table_weighted}), we can provide an explanation for why the FairFN algorithm can achieve such excellent effects: The optimization of the modularity $Q$ in the observed network is decoupled from the optimization of the fairness modularity $Q^P$ in the protected group network. During the iterative process, merging nodes in the unfair direction increases the edge weights between the node and other nodes in the same protected group. If this increase is too large, it can lead to a positive increment in $Q^P$, which will cause the next merge to only select nodes from different protected groups. This negative feedback mechanism ensures that the community remains relatively fair at all times.

\section{In-depth evaluation of the threshold of $\alpha$}
\label{appendix_threshold}
We have also calculated the threshold of $\alpha$ in synthetic clustering datasets and real-world datasets. The results are shown in Figure \ref{fig_point_alpha}. The increase in the threshold of $\alpha$ at the final stage of each figure is very sharp. But in the early stage, the value is really low. There is a transition point that distinguishes these two stages, regardless of whether the degree distribution of vertices follows a power-law distribution (Figure \ref{fig_lfr_alpha}) or a single-point distribution (Figure \ref{fig_point_alpha}). This is somewhat similar to the concept of the phase transition. Actually, this phenomenon also exists in the maximum increment of modularity because the threshold of $\alpha$ is equal to $2m\Delta Q_{max}$. We conjecture that this phenomenon may reflect some important properties of the community partitioning algorithm based on the increment of modularity, which requires further research in the future. 

\begin{figure}[htbp]
    \begin{center}
    \subfloat[Synthetic-equal]{\includegraphics[width=.47\textwidth]{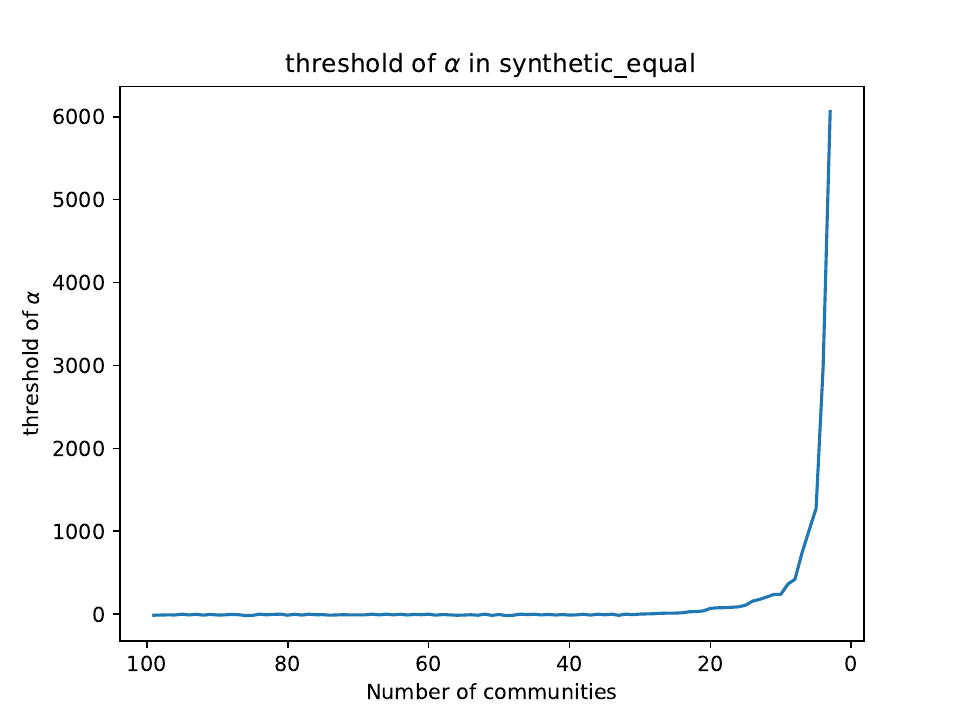}}
    \subfloat[Synthetic-unequal]{\includegraphics[width=.47\textwidth]{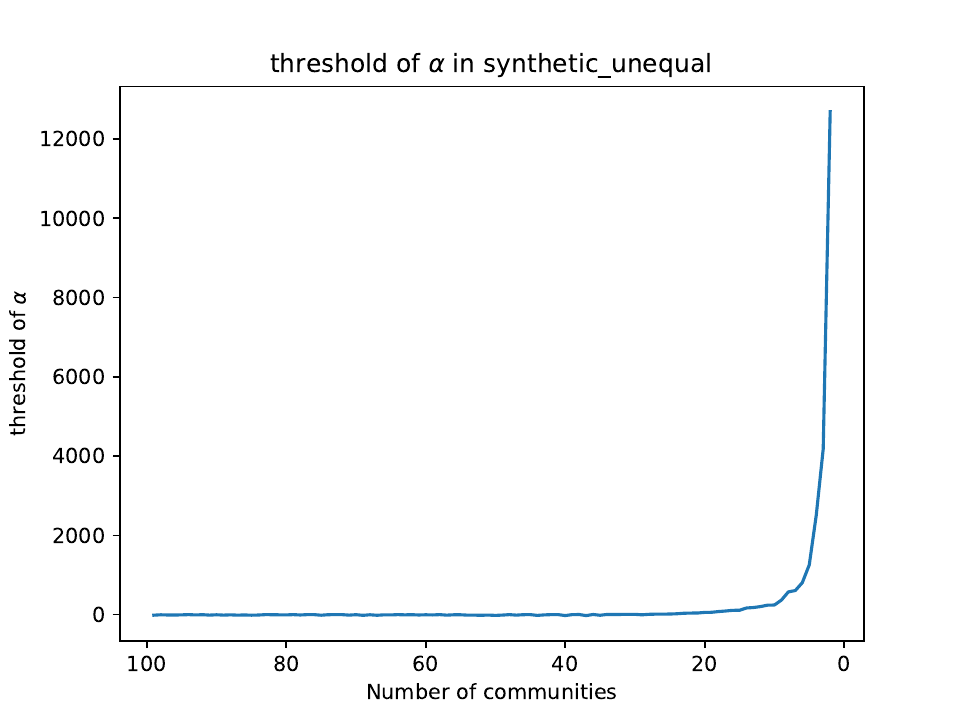}}
    \quad
    \subfloat[Adult]{\includegraphics[width=.47\textwidth]{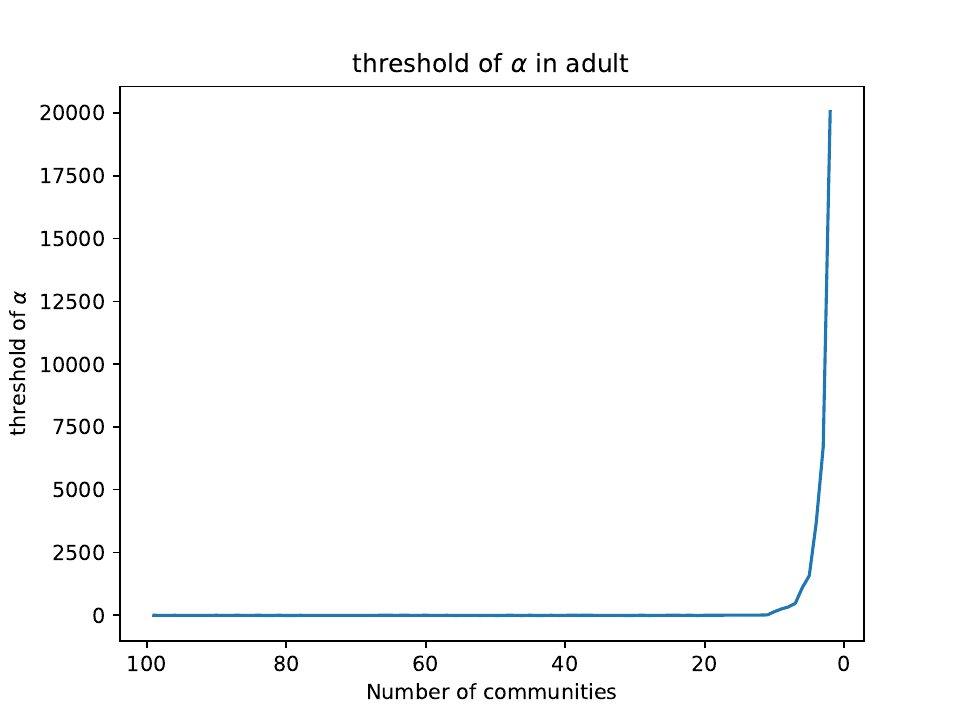}}
    \subfloat[Bank]{\includegraphics[width=.47\textwidth]{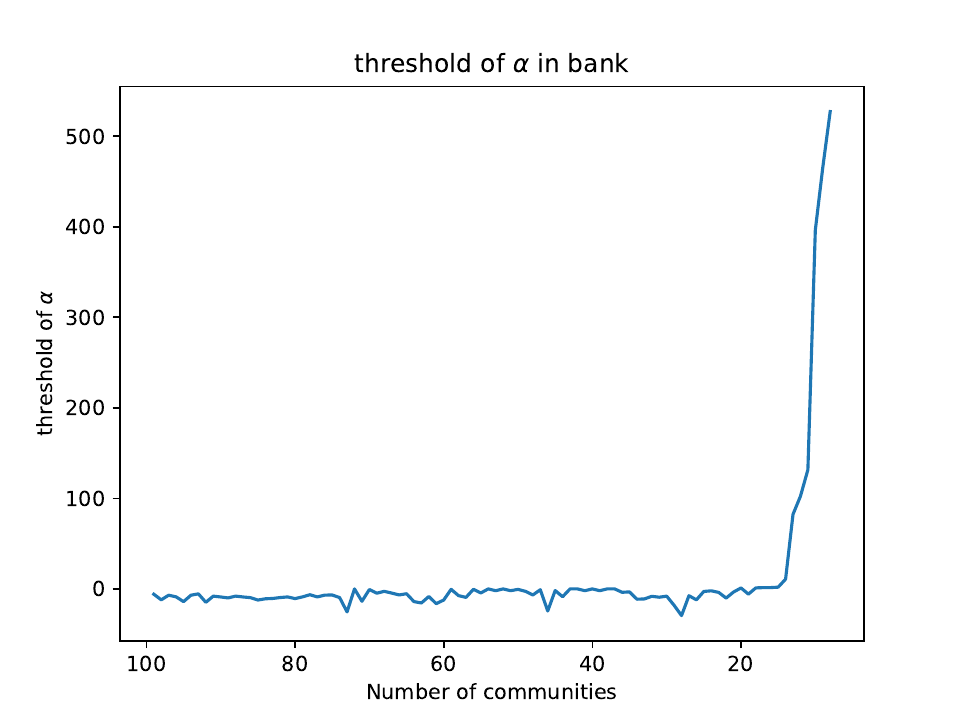}}
    \quad
    \subfloat[Census]{\includegraphics[width=.3\textwidth]{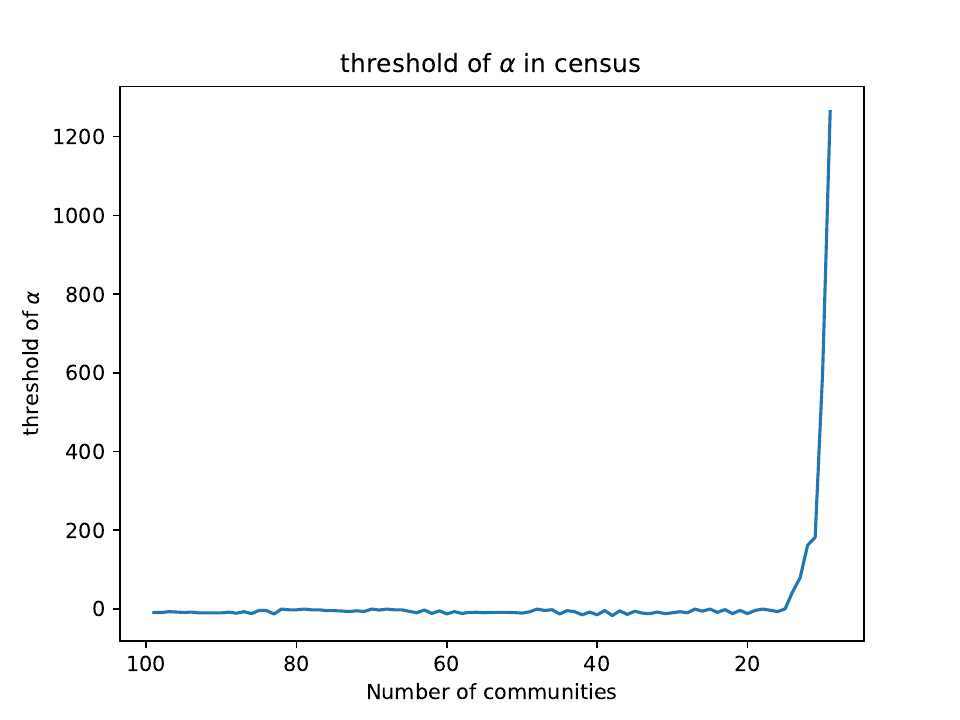}}
    \subfloat[Creditcard]{\includegraphics[width=.3\textwidth]{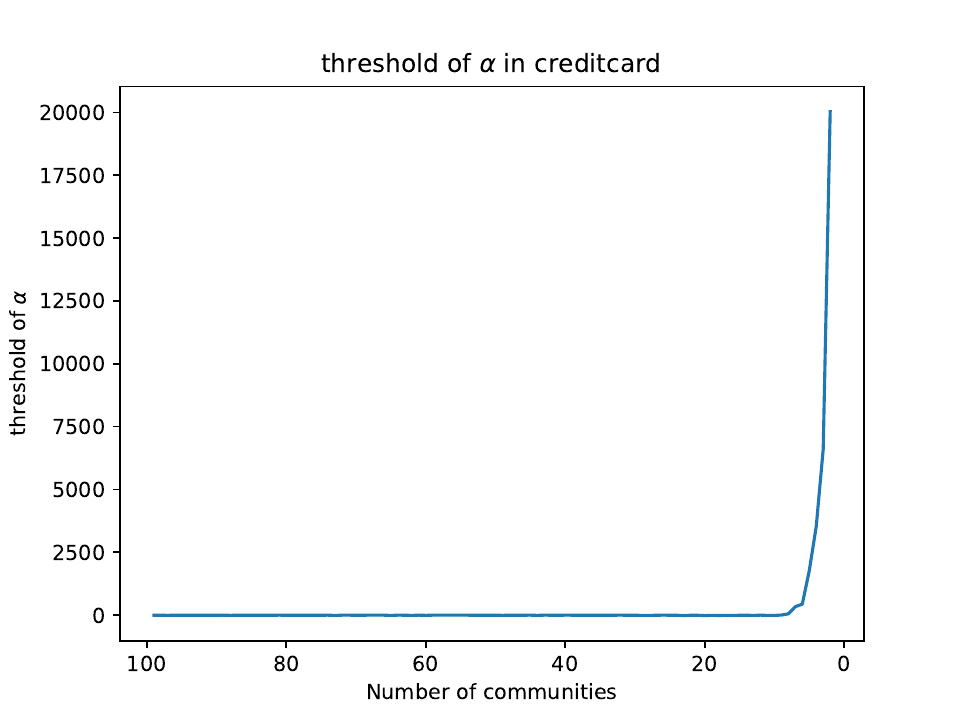}}
    \subfloat[Diabetic]{\includegraphics[width=.3\textwidth]{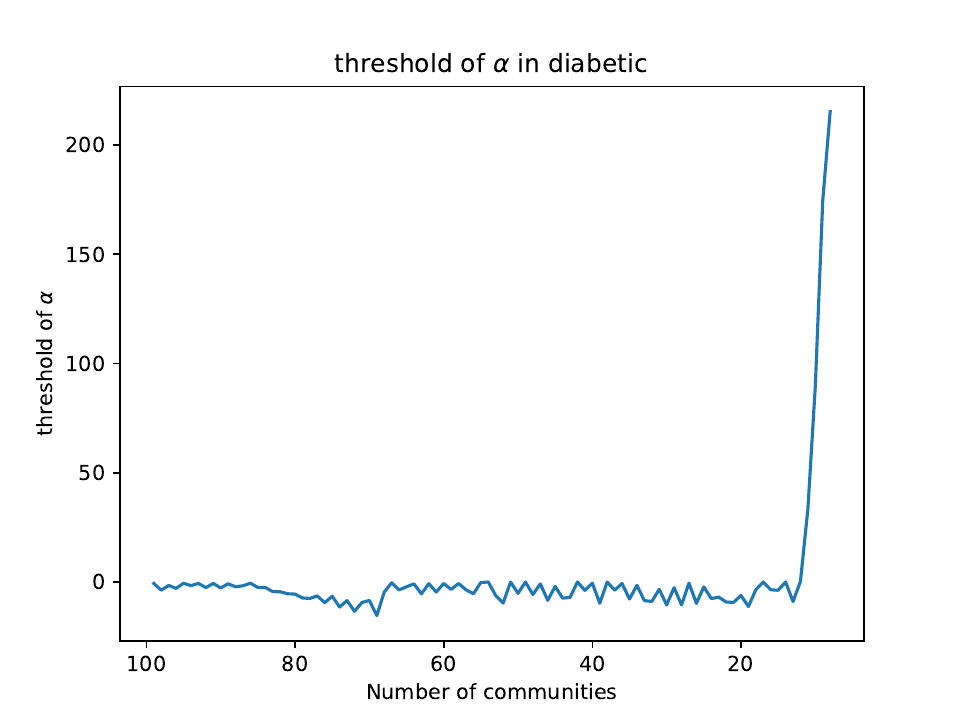}}
    \caption{The threshold of $\alpha$ in synthetic clustering datasets and real-world datasets. The number of communities ranges from 100 to finish due to the threshold in the beginning stage being nearly zero.}
    \label{fig_point_alpha}
    \end{center}
\end{figure}

\section{Broader Impacts}
\label{appendix_impacts}

Fairness-modularity proposed by this paper is very effective and profound in fair community partitioning, which pioneers fair hierarchical clustering in networks. Fairness-modularity and its optimization algorithm will bring broader impacts. In recommendation systems and advertising placement, fair community partitioning helps to reduce bias, break echo chambers, and strengthens communication among different groups. In social network analysis, it helps to ensure the equal right to speak of minority groups and gives better visibility to minority groups. In brain networks, it may help find patterns shared by different individuals, making the results more broadly useful. At the same time, fairness needs to be used carefully. If fairness is applied blindly, it may ignore important social differences. Our method allows users to adjust the fairness level depending on the task. With careful use, this work has the potential to benefit both technology and society in a responsible way.

\end{document}